\documentclass[11pt,a4paper]{article}
\pdfoutput=1
\usepackage{jheppub}

\usepackage[ascii]{inputenc}
\usepackage{dcolumn,bbm,url,booktabs,braket,microtype,aliascnt,amsthm,mathtools,mathrsfs,etoolbox,cancel,cleveref,capt-of,algpseudocode,float,newfloat}
\usepackage[T1]{fontenc}
\usepackage{lmodern}
\usepackage[USenglish]{babel}
\usepackage[normalem]{ulem}

\newtoggle{editing}

\togglefalse{editing}

\iftoggle{editing}{%
  \usepackage[draft]{showlabels}

}{}

\DeclareFloatingEnvironment[fileext=loa,name=Algorithm,placement=th]{algorithm}
\makeatletter
\def\fnum@algorithm{\textbf{\algorithmname\nobreakspace\thealgorithm}}
\makeatother
\newtheorem{thm}{Theorem} 
\newcommand{\newjointcountertheorem}[3]{\newaliascnt{#1}{#2}\newtheorem{#1}[#1]{#3}\aliascntresetthe{#1}}
\newjointcountertheorem{prp}{thm}{Proposition}
\newjointcountertheorem{lem}{thm}{Lemma}
\newjointcountertheorem{cor}{thm}{Corollary}
\newjointcountertheorem{dfn}{thm}{Definition}
\newjointcountertheorem{exl}{thm}{Example}
\newjointcountertheorem{obsdfn}{thm}{Observation and Definition}
\crefname{thm}{theorem}{theorems}
\crefname{prp}{proposition}{propositions}
\crefname{lem}{lemma}{lemmas}
\crefname{cor}{corollary}{corollaries}
\crefname{dfn}{definition}{definitions}
\crefname{exl}{example}{examples}
\crefname{obsdfn}{observation}{observations}
\crefname{figure}{figure}{figures}
\crefname{algorithm}{algorithm}{algorithms}
\crefname{equation}{}{}
\newcommand{\abs}[1]{\lvert#1\rvert}
\newcommand{\RR}{{\mathbb R}}
\newcommand{\calC}{{\mathcal C}}
\newcommand{\norm}[1]{\lVert#1\rVert}
\renewcommand{\autoref}{\usecrefinsteadofautoref}
\renewcommand{\eqref}{\usecrefinsteadofeqref}
\DeclareMathOperator{\AdS}{AdS}
\allowdisplaybreaks[4] 

\title{The Holographic Entropy Cone}

\author[a,b]{Ning Bao,}
\author[c]{Sepehr Nezami,}
\author[b,d]{Hirosi Ooguri,}
\author[b]{Bogdan Stoica,}
\author[e]{James Sully,}
\author[c]{and Michael Walter}

\affiliation[a]{Institute for Quantum Information and Matter, California Institute of Technology, Pasadena, CA 91125, USA}
\affiliation[b]{Walter Burke Institute for Theoretical Physics, California Institute of Technology 452-48, Pasadena, CA 91125, USA}
\affiliation[c]{Stanford Institute for Theoretical Physics, Stanford University, Stanford, CA 94305, USA}
\affiliation[d]{Kavli Institute for the Physics and Mathematics of the Universe, University of Tokyo, Kashiwa 277-8583, Japan}
\affiliation[e]{Theory Group, SLAC National Accelerator Laboratory, Stanford University, Menlo Park, CA 94025, USA}

\emailAdd{ningbao@caltech.edu}
\emailAdd{nezami@stanford.edu}
\emailAdd{ooguri@theory.caltech.edu}
\emailAdd{bstoica@theory.caltech.edu}
\emailAdd{jsully@slac.stanford.edu}
\emailAdd{michael.walter@stanford.edu}

\abstract{%
We initiate a systematic enumeration and classification of entropy inequalities satisfied by the Ryu-Takayanagi formula for conformal field theory states with smooth holographic dual geometries.
For 2, 3, and 4 regions, we prove that the strong subadditivity and the monogamy of mutual information give the complete set of inequalities.
This is in contrast to the situation for generic quantum systems, where a complete set of entropy inequalities is not known for 4 or more regions.
We also find an infinite new family of inequalities applicable to 5 or more regions.
The set of all holographic entropy inequalities bounds the phase space of Ryu-Takayanagi entropies, defining the holographic entropy cone.
We characterize this entropy cone by reducing geometries to minimal graph models that encode the possible cutting and gluing relations of minimal surfaces.
We find that, for a fixed number of regions, there are only finitely many independent entropy inequalities.
To establish new holographic entropy inequalities, we introduce a combinatorial proof technique that may also be of independent interest in Riemannian geometry and graph theory.
}

\preprint{CALT-TH 2015-020 \\ \mbox{}~\hfill IPMU15-0074 \\ \mbox{}~\hfill SLAC-PUB-16294 \\ \mbox{}~\hfill SU-ITP-15/08}

\begin{document}
\maketitle

\section{Introduction and summary of results}

Understanding the microscopic mechanism of holography is one of the outstanding problems in quantum gravity.
In the context of the AdS/CFT correspondence, there have been hints that quantum entanglement in the boundary CFT plays a role in emergence of the bulk AdS geometry \cite{mark,swingle}.
Chief among these hints is the Ryu-Takayanagi proposal \cite{rt06,rt06long},\footnote{See \cite{2007JHEP07062H} for the covariant generalization.} which states that the entanglement entropy of a spatial region $A$ in a holographic CFT is given by
\begin{equation}
\label{eq:ryu takayanagi}
  S(A) = \min_{A'} \frac {\abs{A'}} {4 G_N},
\end{equation}
to leading order in the central charge, where the minimization is over all hypersurfaces $A'$ in the bulk time slice homologous to $A$.

There have been many profound implications of this formula.
Notably, when the formula is applied to the case of the two-sided black hole geometry, one finds that the Einstein-Rosen bridge encodes entanglement of thermofield double degrees of freedom in the boundary CFT \cite{eternal,maldacenahartman}.
Moreover, the Ryu-Takayanagi formula has been used to derive the linearized Einstein equations in AdS \cite{raamsdonketal,swingle}, as well as to obtain integrated energy conditions \cite{2014arXiv1412.1879L,Lashkari:2014kda}.
These examples all demonstrate important connections between bulk geometry and entanglement data on the boundary.

Not all CFT states have holographic duals with a smooth geometry.
Since the bulk geometry reflects boundary entanglement, we envisage that criteria for smooth bulk geometry can be partially expressed in terms of entanglement properties of the boundary CFT.
The situation is similar to the hydrodynamic description of many body systems, where the evolution of a system is approximately described in terms of emergent macroscopic degrees of freedom; the validity of the approximation can be stated in terms of properties of macroscopic observables such as density and velocity.
In this sense, holography is akin to a hydrodynamic description of the boundary entanglement with entropies as its macroscopic phase space.

We wish to identify the criteria that bound this entropic phase space for holographic states.
It has been a long-standing open question to determine whether there are any further universal entropy inequalities apart from the classical ones \cite{Pippenger86,ZhangYeung97,Pippenger03}.
For probability distributions, this question has been answered to the affirmative in the breakthrough works \cite{ZhangYeung98,Matus07}: For $n \geq 4$ random variables, there is an infinite number of independent entropy inequalities satisfied by the Shannon entropy.
In network coding theory, they give rise to tighter capacity bounds \cite{nonshannon}.
The quantum case has so far remained elusive, though there has been partial progress \cite{LindenWinter05,CadneyLindenWinter12,LindenMatusRuskaiEtAl13,GrossWalter13}.
Understanding the phase space of entanglement entropies for specific subclasses of quantum systems is also of general interest for information theorists.
Holographic states are an interesting set of quantum states, and are in many ways similar to the set of all stabilizer states.
In both cases, special restrictions on entanglement entropy 
allow for the usage of more powerful methods in attempting to characterize 
the inequalities that govern them.
For example, it is known that the Ryu-Takayanagi formula for holographic entanglement entropy satisfies the following monogamy inequality for the mutual information, which is defined by $I(A:B) = S(A) + S(B) - S(AB)$:
\begin{equation}
\label{eq:monogamy}
  I(A:BC) \geq I(A:B) + I(A:C),
\end{equation}
for any three disjoint regions $A$, $B$, $C$ \cite{hayden_headrick_maloney_2013}.
This inequality is not necessarily satisfied by generic quantum systems and therefore distinguishes those with smooth holographic duals.
It also has profound implications for the structure of holographic quantum states and their correlations.
For example, it suggests that holographically any quantum Markov chain is trivial and entropically excludes higher GHZ states \cite{hayden_headrick_maloney_2013,multiboundary_2014}.
All these observations raise the natural question of whether there are further entropic constraints required of holographic theories.
Indeed, any such entropic constraint can be understood as a necessary condition for the existence of a smooth holographic dual.

\medskip

In this paper, we initiate a systematic enumeration and classification of entropy inequalities satisfied by the Ryu-Takayanagi formula.
We will begin our study by defining the notion of the \emph{holographic entropy cone}.
For a given number of boundary regions, the holographic entropy cone parametrizes the phase space of allowed Ryu-Takayanagi entropies in arbitrary holographic CFTs.
Akin to the classical case, we show that the phase space is closed under suitably-defined addition and scaling, giving it the structure of a convex cone.
Our entropy inequalities form the facets of this cone, but it can equally be described by its extreme rays, the one-dimensional intersections of these inequalities.
These rays are extremal in the entropic phase space, and they have a privileged role in our considerations.

The framework of entropy cones allows us to study holographic entropies in a systematic fashion by using convex geometry.
We prove that, for 2, 3 and 4 regions, the monogamy of mutual information \cref{eq:monogamy} and strong subadditivity,
\[  S(AB) + S(BC) \geq S(B) + S(ABC), \]
give the complete set of inequalities.
This is in contrast to the situation for generic quantum systems, where a complete set of entropy inequalities is not known for 4 or more regions.
The pattern changes at 5 regions, where we find a new set of inequalities.
As we increase the number of regions, a new infinite family of independent entropy inequalities emerges, of which
the strong subadditivity and the monogamy of mutual information are two special cases.
This family of inequalities can be understood as a bound of the form proposed for differential entropy \cite{bartek}, but of much more general application.

The systematic study of holographic entropy has required the development of new tools, which are of considerable interest in their own right.
While a continuous geometry describes infinitely many minimal surfaces, any entropy inequality is only concerned with a finite number of boundary regions.
We show that the relevant information in the geometry can thus be reduced to a finite, weighted graph.
This \emph{graph model} is a minimal model for the geometry that faithfully reproduces all Ryu-Takayanagi entropies by a simple graph-theoretic prescription.
It also contains information on all possible ways the minimal surfaces can be cut and reassembled to generate new surfaces that bound the entanglement entropies of other bulk regions.
Graph models allow us to study holographic entropies using combinatorial methods.
One insight that follows is that, for a fixed number regions, there are only finitely many independent inequalities required by the Ryu-Takayanagi entropy formula; in other words, that the holographic entropy cones are \emph{polyhedral}.

We would also like to understand what characterizes the states that only marginally satisfy these constraints.
These states lie at the extremal boundary between the states that are realizable holographically and those that are not.
If entanglement is somehow constitutive of spacetime, then we expect extremal states to be geometries that are in some sense nearly torn apart.
It has similarly been suggested that EPR pairs create an ER bridge \cite{lenny} and that, with sufficient amounts of properly organized entanglement, these wormholes coalesce into smooth geometry.
Therefore, it may be reasonable to expect that the extremal geometries are themselves wormhole geometries.

From the perspective of entropy, the extremal geometries correspond to extreme rays of the holographic entropy cone.
The entropies for an arbitrary smooth geometry can in turn be obtained by specific convex combinations of the extreme rays.
This can be understood as making visible the individual threads of the entropy fabric.
We show that the entropies of any extreme ray can indeed be explained by certain extremal multiboundary wormhole geometries as considered in \cite{skenderis_van_rees_2011,Brill}.
In fact, we prove more generally that any graph model naturally gives rise to a multiboundary wormhole geometry that has the same entropies.
This shows that graph models provide a completely equivalent, combinatorial description of holographic entropy.

At last, graph models are also the starting point for systematically generalizing the inclusion/exclusion proofs in \cite{headrick_takayanagi_2007,hayden_headrick_maloney_2013}, which had been used to establish the strong subadditivity and monogamy of the mutual information.
To prove new inequalities, we furthermore introduce the concept of \emph{proofs by contraction}.
A proof of a holographic entropy inequality of the form $LHS \geq RHS$ is given by a cutting of the Ryu-Takayanagi surfaces of the $LHS$ and gluing of a subset of the resulting pieces into bounding surfaces for the Ryu-Takayanagi surfaces of the $RHS$.
We show that suitable maps on hypercubes describe this process in a purely combinatorial fashion; the ability to glue segments to obtain the $RHS$ is formalized as a \emph{contraction} property.
Using this criterion, we prove a new set of inequalities, revealing a hitherto unknown rich spectrum of entropic constraints in holographic theories.

\paragraph{Organization of the paper.}

In \cref{sec:entropy cone}, we define the holographic entropy cone, motivated by ideas from quantum information theory.
We discuss several examples and derive a number of general properties.
In \cref{sec:graph model}, we introduce the graph model for bulk geometries.
Together with a discrete variant of the Ryu-Takayanagi formula, we obtain a purely combinatorial description of the holographic entropy cone. 
In \cref{sec:inequalities} we introduce proofs by contraction, our method for proving holographic entropy inequalities that generalize the inclusion/exclusion arguments of \cite{headrick_takayanagi_2007,hayden_headrick_maloney_2013}.
We proceed to prove a rich new set of inequalities.
Finally, in \cref{sec:physics} we discuss the physical implications of our work.
We conclude in \cref{sec:outlook} by sketching several interesting avenues for future investigation.

\section{The holographic entropy cone}
\label{sec:entropy cone}

In this work, we are interested in the inequalities satisfied by entanglement entropies of states with holographic description
in terms of smooth bulk geometries.
Just as with Shannon and von Neumann entropy, it will be useful to define the notion of an entropy cone \cite{ZhangYeung97,Pippenger03}.

Let $A_1, \ldots, A_n$ be $n$ {\it disjoint} regions in the boundary field theory.
For any non-empty subset $I \subseteq [n]$\footnote{We employ the notation $[n]:= \{1,\dots,n\}$.}, we may compute the entropy $S(I) := S(A_I)$ of the corresponding composite boundary region $A_I = \bigcup_{i \in I} A_i$ by using the Ryu-Takayanagi prescription \cref{eq:ryu takayanagi}.
When $A_i$ itself has a boundary, its entanglement entropy is UV divergent and we need to
to choose a cut-off; geometrically, this corresponds to cutting off the bulk time slice as to obtain a compact manifold with boundary.
By listing the entropies of all these subsystems one after another, we obtain a \emph{holographic entropy vector} $(S(A_I))_{\emptyset \neq I \subseteq [n]} \in \RR^{2^n-1}$.
For example, the entropy vector of a bipartite state is given by the triple $(S(A), S(B), S(AB))$, where $A$ and $B$ denote the boundary regions (cf.~\cref{fig:cone2}).

\begin{figure}
\centering
\includegraphics[width=0.5\linewidth]{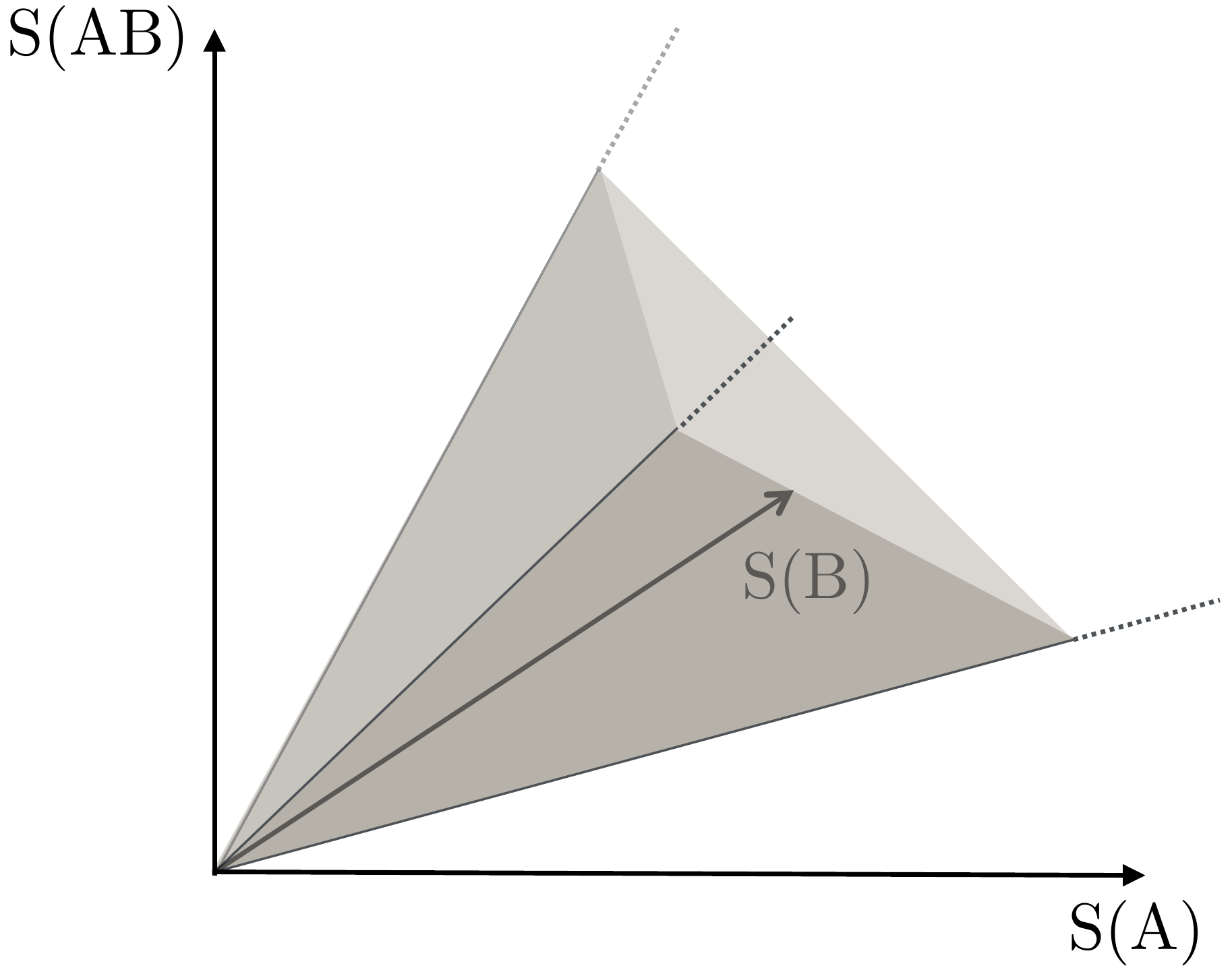}
\caption{The holographic entropy cone $\calC_2$ for two regions. See \cref{subsec:few regions} for a discussion of its facets and extreme rays.}
\label{fig:cone2}
\end{figure}

Mathematically, the Ryu-Takayanagi formula makes sense for an arbitrary bulk manifold $X$, regardless of the existence a dual CFT state on its boundary, and it will be convenient to allow for this generalization.
Thus we \emph{define} the holographic entropy of a boundary region $A \subseteq \partial X$ of an arbitrary bulk manifold $X$ by \cref{eq:ryu takayanagi}.
It is easy to see the following:

\begin{obsdfn}
\label{obs:convex cone}
The set $\calC_n$ of all holographic entropy vectors, obtained by varying the boundary regions $A_1, \dots, A_n \subseteq \partial X$ as well as the bulk manifold $X$, is a convex cone, which we will call the \textbf{holographic entropy cone} of $n$ regions.
\end{obsdfn}
\begin{proof}
  To see that $\calC_n$ is closed under multiplication by positive scalars, observe that rescaling the Riemannian metric rescales all entropies accordingly.

  To see that it is closed under addition, consider two manifolds $X$, $X'$ and boundary regions $A_1,\dots,A_n \subseteq \partial X$, $A'_1,\dots,A'_n \subseteq \partial X'$ and construct the disjoint union $X \sqcup X'$ together with the boundary regions $A''_i := A_i \sqcup A_i'$. Then the corresponding entropy vector is given by the sums $S(A''_I) = S(A_I) + S(A'_I)$.
\end{proof}

Remarkably, we will see later that this generalization to arbitrary bulk manifolds is in fact immaterial. We will show in \cref{sec:graph model} that any entropy vector obtained from an arbitrary bulk manifold can always be explained by a hyperbolic surface (\cref{thm:holographic cone}), and we will argue in \cref{sec:physics} that any hyperbolic surface -- more generally, any Riemannian manifold with constant scalar curvature and and asymptotically locally hyperbolic boundaries -- is holographically dual to a CFT state on its boundaries such that that the Ryu-Takayanagi formula \cref{eq:ryu takayanagi} gives the entanglement entropies.

\Cref{obs:convex cone} implies that, up to taking its topological closure, $\calC_n$ can be described by linear inequalities.
We will later see that $\calC_n$ is not only already closed but in fact a \emph{polyhedral} convex cone (see \cref{prp:polyhedral} below).
Therefore, $\calC_n$ is cut out by a finite number of linear entropy inequalities \[ \sum_{I \in [n]} c_I S(I) \geq 0 \] that hold for the holographic entropies defined by any $n$ boundary regions of an arbitrary bulk geometry $X$.
On a conceptual level, the convexity of the entropy cone justifies the search for \emph{linear} entropy inequalities.
Geometrically, the facets of the holographic entropy cone can be identified with a minimal set of entropy inequalities defining the cone: Any vector $(s_I) \in \RR^{2^n-1}$ that satisfies these inequalities can be realized by the Ryu-Takayanagi entropy formula for some choice of bulk geometry and boundary regions.

Like any polyhedral cone, the holographic entropy cone can alternatively be defined in terms of finitely many \emph{extreme rays}.
Those are the rays in $\calC_n$ that cannot be written as a proper convex combination of other elements in $\calC_n$.
This dual perspective will be useful in \cref{subsec:few regions} for demonstrating that a set of entropy inequalities is complete.
From a physical point of view, the extreme rays can be seen as the entropic building blocks from which an arbitrary entropy vector can be constructed by convex combination.

\subsection{Transformations and symmetries}
\label{subsec:symmetries}

Recall that any quantum state can be purified to a pure state on a larger system.
Similarly, we may always add the region $A_{n+1} = \partial X \setminus \bigcup_{i=1}^n A_i$, so that $\bigcup_{i=1}^{n+1} A_i = \partial X$.%
\footnote{$\partial X$ is the topological boundary of the bulk manifold $X$, which need not be the same as the conformal boundary; black hole horizons (if present) are also included in $\partial X$. In particular, we do not assume that $A_{n+1}$ corresponds to a region in the conformal field theory.}
We will call $A_{n+1}$ the \emph{purifying region} and oftentimes also denote it by $O$.
We note that, with the purifying region added, the entropy of any boundary region then agrees with the entropy of its complement.
This is because any Ryu-Takayanagi surface $A'$ for a boundary region $A$ is simultaneously a Ryu-Takayanagi surface for $A^c=\partial X \setminus A$:
If $a$ is a bulk region such that $\partial a = A \cup A'$ then $\partial (X\setminus a) = A^c \cup A'$.
In particular, $S(I) = S(I^c)$ for all $I \subseteq [n+1]$.
Our construction therefore reproduces the entropies of any purification of a CFT state corresponding to the bulk geometry $X$ and boundary regions $A_1, \dots, A_n$ (independent of whether this purification is geometric or not).
In this way, we obtain an \emph{embedding} $\calC_n \rightarrow \calC_{n+1}$ of holographic entropy cones.
Conversely, forgetting the region labeled $n+1$ yields a \emph{projection} $\calC_{n+1} \rightarrow \calC_n$.

It is clear that the holographic entropy cone is left invariant by relabeling the boundary regions $A_1, \dots, A_n$ of the bulk manifold $X$.
Using the above operations, we can extend this to an action of the permutation group $S_{n+1}$ by relabeling the $n+1$ boundary regions obtained by adjoining the purifying region $A_{n+1}$.
\footnote{Explicitly, if $\pi$ is a permutation of $[n+1]$ then its action on an entropy vector $s = (s_I)_{\emptyset \neq I \subseteq [n]} \in \calC_n$ is given
by replacing the entropy of subsystem $I = \{i_1, \ldots, i_k\}$ by the entropy of $\{\pi^{-1}(i_1), \dots, \pi^{-1}(i_k)\}$ or its complement in $[n+1]$, depending on which subsystem is contained in $[n]$.}
This extended \emph{permutation symmetry} is highly useful in revealing the combinatorial structure of the holographic entropy cones.
For example, it is well-known that \emph{strong subadditivity}, $S(AB) + S(BC) \geq S(B) + S(ABC)$, and \emph{weak monotonicity}, $S(AB) + S(BC) \geq S(A) + S(C)$, can be identified via this symmetry even though they correspond to different faces of the entropy cone.
The extended permutation symmetry can also be used to obtain several variants of the embedding and projection defined above:
For example, we can embed $\calC_n \rightarrow \calC_{n+1}$ by adding an empty boundary region $A_{n+1} = \emptyset$, and there exist coarse-graining projections $\calC_{n+1} \rightarrow \calC_n$ given by combining any two regions (e.g., $A'_n = A_n \cup A_{n+1}$).
We refer to \cite{majenz14} for a lucid discussion of the morphisms of the entropy cone for von Neumann entropy.

By duality, we may also apply entropy inequalities for $n$ to $n+1$ regions and conversely.
For example, the entropy inequality $I(A:B) \geq 0$ for two regions can be applied to two regions of a three-party system, but also in the form $I(A:BC) \geq 0$.
And the strong subadditivity $S(AB) + S(BC) \geq S(B) + S(ABC)$ reduces to both \emph{subadditivity} $S(A) + S(C) \geq S(AC)$ and the Araki-Lieb inequality $S(AB) \geq S(A) - S(B)$, depending on whether we set the $B$ system to be trivial or apply strong subadditivity to the purification.

\subsection{Fewer than five regions}
\label{subsec:few regions}

We will now compute the holographic entropy cones for $n \leq 4$ regions.
Recall that it had previously been shown that the Ryu-Takayanagi formula satisfies \emph{strong subadditivity} \cite{headrick_takayanagi_2007},
\begin{equation}
\label{eq:ssa}
  S(AB) + S(BC) \geq S(B) + S(ABC),
\end{equation}
and that the holographic mutual information is monogamous \cite{hayden_headrick_maloney_2013}, which can be written as
\begin{equation}
\label{eq:mmi}
  S(AB) + S(BC) + S(AC) \geq S(A) + S(B) + S(C) + S(ABC).
\end{equation}
The proofs are purely geometrical in nature and in particular hold for arbitrary bulk manifolds.
While the above inequalities refer to three boundary regions $A$, $B$ and $C$, we may use the operations described in \cref{subsec:symmetries} to obtain from \cref{eq:ssa} and \cref{eq:mmi} new inequalities for different numbers of regions.

Our strategy for computing the holographic entropy cone then is the following:
Given a number of regions $n$, we first compute the extreme rays of the convex cone $\widehat\calC_n$ cut out by all inequalities obtained in the above way.
Since all these inequalities are satisfied by holographic entropies, the cone $\widehat\calC_n$ contains the holographic entropy cone $\calC_n$.
For each extreme ray of $\widehat\calC_n$, we then try to find a bulk geometry and boundary regions such that the associated entropy vector lies on this ray.
If we succeed in doing so then convexity implies that, in fact, $\widehat\calC_n = \calC_n$. 
We have successfully implemented this strategy for $n \leq 4$ boundary regions.
That is, in each case, we find that the holographic entropy cone is cut out by strong subadditivity \cref{eq:ssa} and monogamy of the mutual information \cref{eq:mmi} (up to symmetry).
In fact, the only instances of \cref{eq:ssa} that are required are the \emph{subadditivity} inequality for single regions, $S(A) + S(C) \geq S(AC)$, and its permutations.
In particular, \emph{strong} subadditivity does not correspond to facets of $\calC_n$ but rather follows as a consequence of \cref{eq:ssa,eq:mmi}.
We now describe the extreme rays and associated bulk geometries that we have found.

\begin{figure}
\centering
\includegraphics[width=0.8\linewidth]{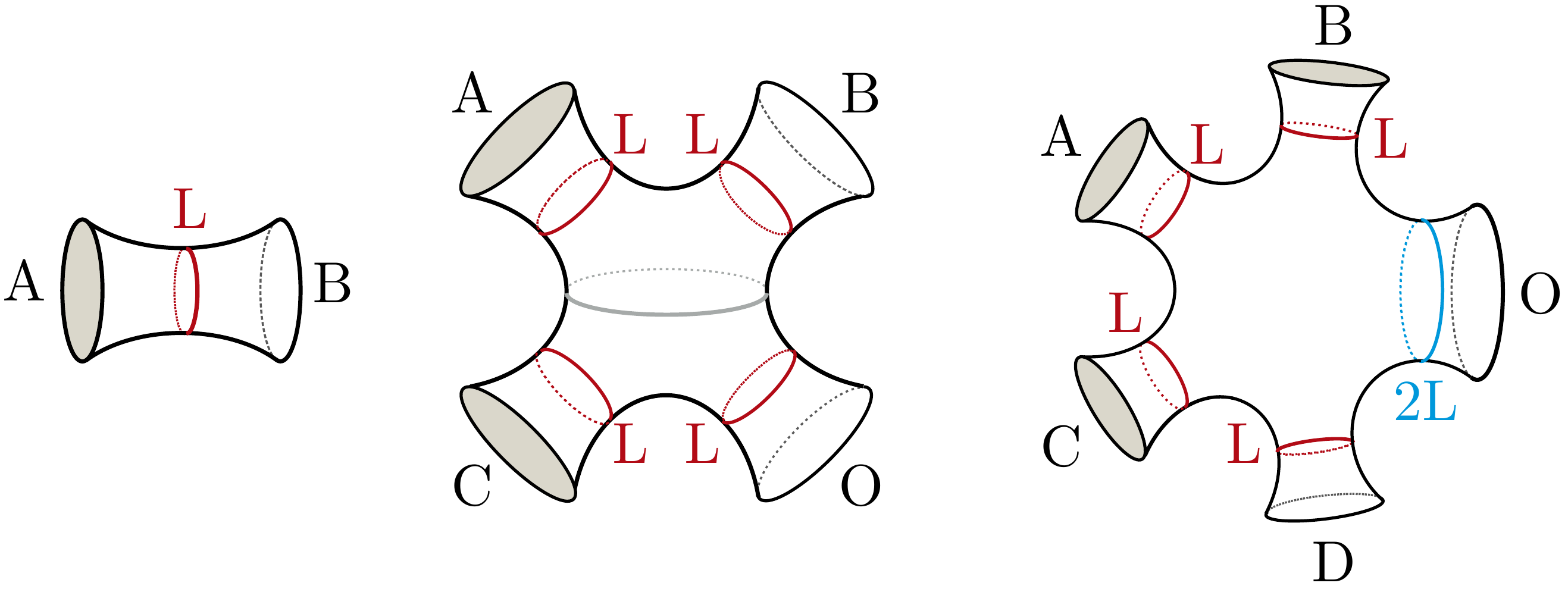}
\caption{Geometries realizing the extreme rays of the holographic entropy cones for $n \leq 4$ regions.}
\label{fig:rays234}
\end{figure}

\paragraph{Two regions.} Up to symmetry, there is a single extreme ray, corresponding to the entropies of an EPR pair shared between two systems $A$ and $B$,
\begin{equation}
\label{eq:ray2}
  (s_A, s_B; s_{AB}) = (1,1;0).
\end{equation}
In holography, this extreme ray can naturally be realized by taking the thermofield double state on two CFTs, which corresponds geometrically to the time-reflection symmetric slice of a `wormhole' or Einstein-Rosen bridge that connects two asymptotic AdS regions \cite{lenny}.
We refer to the first cartoon in \cref{fig:rays234} for an illustration of the geometry thus constructed.
The two other extreme rays can be obtained by applying the $S_3$-permutation symmetry to the ray \cref{eq:ray2}.
See \cref{fig:cone2} for an illustration of the entropy cone thus obtained.

\paragraph{Three regions.} There are two extreme rays up to symmetry. The first class of rays is inherited from $\calC_2$ and corresponds to Bell pairs shared between individual subsystems.
The second ray is new and given by
\begin{equation}
\label{eq:ray3}
  (s_A,s_B,s_C; s_{AB},s_{AC},s_{BC}; s_{ABC}) = (1,1,1;2,2,2;1).
\end{equation}
Geometrically, we can realize this ray by a hyperbolic surface with four geodesic boundaries $A$, $B$, $C$ and $O$ of equal length $L$.
Note that $O$ is the purifying region.
If we choose $L$ sufficiently small then any internal cycle is longer than twice this length. In this case, the Ryu-Takayanagi formula gives $S(A) = L$ for any single boundary, $S(AB) = 2L$ for any pair of boundaries, and $S(ABC) = S(O) = L$ since $O$ is the purifying region.
We refer to the second cartoon in \cref{fig:rays234} for an illustration of this geometry (the red cycles denote the boundary geodesics and the gray cycle is one of the many irrelevant internal cycles).
The surface thus constructed can be seen as a time slice of a multiboundary wormhole geometry.
We will see in \cref{sec:graph model,sec:physics} that this is no accident but rather reflects a general feature of holographic entropy.

\paragraph{Four regions.} There are three extreme rays up to symmetry. The first two rays are inherited from $\calC_3$ and can therefore be explained by the geometries that we have described before. A new ray is given by the following assignment of entropies:
\begin{equation}
\label{eq:ray4}
\begin{aligned}
    &(s_A,s_B,s_C; s_{AB},s_{AC},s_{AD},s_{BC},s_{BD},s_{CD}; s_{ABC},s_{ABD},s_{ACD},s_{BCD}; s_{ABCD}) \\
  =\;&(1,1,1,1; 2,2,2,2,2,2; 3;3;3;3; 2).
\end{aligned}
\end{equation}
It can be realized geometrically by considering a five-boundary wormhole geometry, where all boundaries have small length $L$ except for the purifying boundary, which has twice that length. We refer to the third cartoon in \cref{fig:rays234} for an illustration of this geometry.

\paragraph{Five and more regions.}
We will see in \cref{subsec:five} that the situation is markedly different for five and more regions.
Here, there are in general many other inequalities independent from subadditivity and monogamy, and the geometries corresponding to extreme rays are no longer of the simple form that we found above.

\paragraph{Other entropy inequalities.}
It is a highly non-trivial result that strong subadditivity \cref{eq:ssa} is valid for arbitrary density matrices \cite{LiebRuskai73}.
While strong subadditivity and its permutations characterize the quantum entropy cone for $n\leq3$ completely \cite{Pippenger03}, there are several other inequalities for $n \geq 4$ that are conjectured to hold, partly due to substantial numerical evidence \cite{ibinson_2006,walter_2014}. One example is the Zhang-Yeung inequality which is known to hold for the Shannon entropy of four random variables \cite{ZhangYeung98}:
\begin{equation}
  \label{eq:zhang yeung}
  2 I(C:D) \leq I(A:B) + I(A:CD) + 3 I(C:D|A) + I(C:D|B).
\end{equation}
It is interesting to note that \cref{eq:zhang yeung} holds for holographic entropies, as it is a trivial consequence of \cref{eq:ssa} and \cref{eq:mmi}, which can be written as $I(A:C|B) \geq 0$ and $I(A:B) \leq I(A:B|C)$, respectively.

On the other hand, it is not hard to see that the monogamy inequality \cref{eq:mmi} is not valid for general quantum states.
Another such example is the Ingleton inequality \cite{ingleton_1971}, which can similarly be seen to hold for holographic entropies but not for general quantum states:
\begin{equation}
  \label{eq:ingleton}
  I(A:B|C) + I(A:B|D) + I(C:D) \geq I(A:B)
\end{equation}
It has previously been shown that the entropies of \emph{stabilizer states}, which are an important class of quantum error correcting codes, likewise satisfy the Ingleton inequality \cite{LindenMatusRuskaiEtAl13,GrossWalter13,walter_2014}.
In particular this implies that all four-partite holographic entropies can be explained by (mixtures of) stabilizer entropies, since it is known that, for four subsystems, \cref{eq:ingleton} is the only additional linear constraint for stabilizer states \cite{LindenMatusRuskaiEtAl13}.

In fact, the extreme rays of the holographic entropy cone $\calC_4$ form a \emph{proper} subset of the extreme rays of the entropy cone spanned by stabilizer states. It is easy to write down stabilizer states corresponding to each extreme ray: The ray \cref{eq:ray2} can be realized by an EPR pair, the ray \cref{eq:ray3} by the four-qutrit stabilizer state $\sum_{i,j=0}^2 \ket{i,j,i+j,i+2k}$, and the ray \cref{eq:ray4} by the well-known five-qubit stabilizer code \cite{ibinson_2006,LindenMatusRuskaiEtAl13}.
This is rather interesting, as it indicates that the quantum error-correcting codes constructed from black holes are additionally constrained \cite{almheiri_dong_harlow_2014}.
For example, the four-partite GHZ state $\ket{0000} + \ket{1111}$ is a stabilizer state that violates \cref{eq:mmi} and so \emph{cannot} be realized holographically.
The connection between holographic entropies and stabilizer entropies is also conceptually pleasing since stabilizer states -- just like states of holographic conformal field theories -- have a classical dual description (in terms of Wigner functions on a classical phase space, see e.g. \cite{GrossWalter13}).

\paragraph{Uniqueness.}
We remark that the choice of geometry realizing an extreme ray is in general not unique.
For example, the ray \cref{eq:ray2} can also be realized by taking $A$ and $B$ as two complementary non-empty regions of the boundary of any bulk manifold.

On the other hand, spacetimes with multiple boundaries are the natural setting for extreme ray geometries.
In fact, there are extreme rays that cannot be realized as holographic entropy vectors
of geometries with a single connected boundary.
Let us consider here a specific example, namely that of the extreme ray \cref{eq:ray3}, where $S_A=S_B=S_C=S_{AB}/2=S_{AC}/2=S_{BC}/2=S_{ABC}$. While it is possible to get the finite pieces of the entanglement entropies to match in this fashion, matching the divergent pieces \emph{as well} becomes much trickier. The requirement that the 1-body entropies are the same enforces that $A$, $B$ and $C$ all require the same number of boundaries, and the requirement that the 2-body entropies are twice the 1-body entropies requires that no two intervals that are part of $A$, $B$, or $C$ are adjacent to each other. However, this immediately forces the number of boundaries of $ABC$ to be three times that of $A$, meaning that there is no possible way to match the 3-body entropy with the 1-body entropy, unless we take the limit where the cutoff is of the same scale as the interval size, allowing for unphysical competition between the finite and divergent pieces.

Even in cases where extreme rays are realizable on a single boundary, the entropy vector should be understood to be capturing the behavior of our choice of cut-off.
Accounting for the purification, these single-boundary vectors always contain adjacent regions.
Their spatial connectivity is precisely a consequence of divergences in their correlators and mutual information. 

\section{The graph model of holographic entropies}
\label{sec:graph model}

In this section, we introduce an alternative combinatorial model for holographic entropies and use it to reveal some properties of the holographic entropy cone.

\subsection{From geometries to graphs}
\label{subsec:tograph}

Let $X$ be an arbitrary bulk manifold and $A_1, \dots, A_n$ disjoint boundary regions.
In \cref{sec:entropy cone}, we have defined the associated holographic entropy vector as the collection of Ryu-Takayanagi entropies of all composite regions $A_I = \bigcup_{i \in I} A_i$, where $\emptyset \neq I \subseteq [n]$.
For each such region, let us denote a corresponding Ryu-Takayanagi surface $A'_I$, i.e., a minimizer of \cref{eq:ryu takayanagi}.
Together, these $2^{n}-1$ surfaces cut the bulk geometry into a finite number of connected pieces, as illustrated in \cref{fig:tograph-1}.
\begin{center}
  \includegraphics[width=0.55\linewidth]{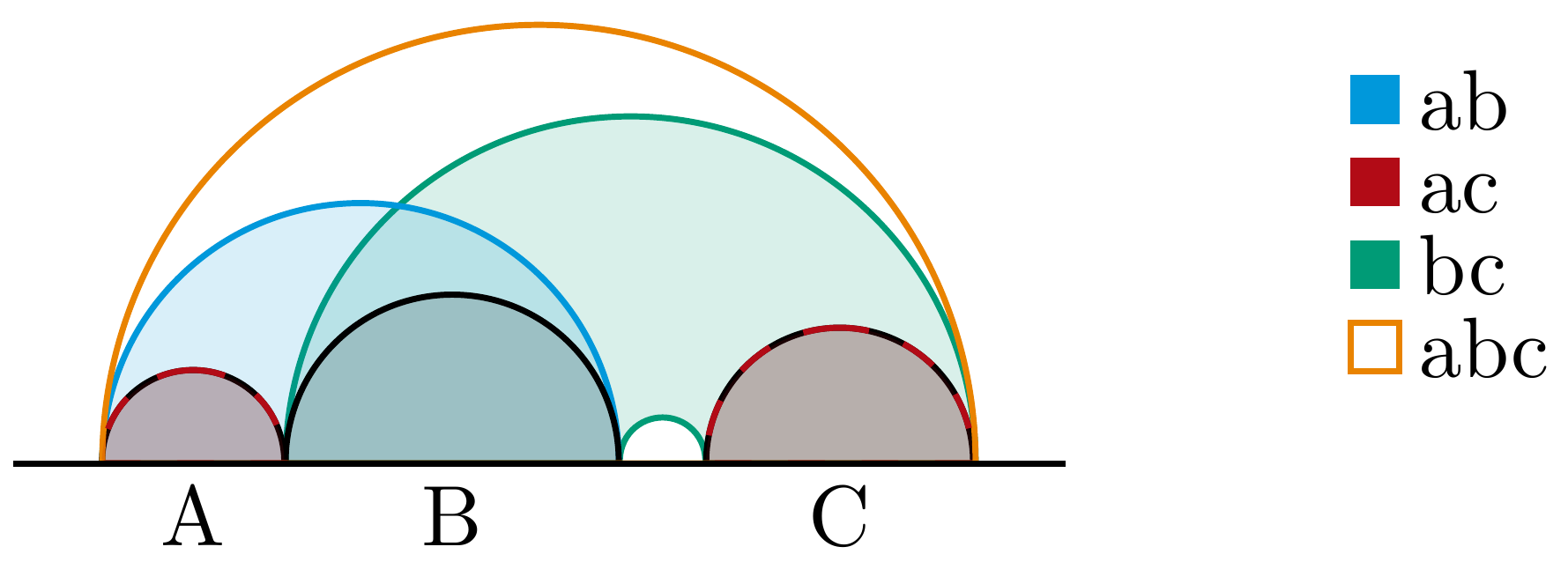}
  \captionof{figure}{The Ryu-Takayanagi surfaces of the $2^n-1$ subsystems cut the bulk geometry into a finite number of pieces.}
  \label{fig:tograph-1}
\end{center}
Formally, if $a_I$ is a bulk region that implements the cobordism between $A_I$ and its Ryu-Takayanagi surface, i.e., $\partial_I a_I = A_I \cup A'_I$, then each piece is given by a connected component of an intersection of regions $a_I$ and their complements $a_I^c$.

We now define a graph by adding one vertex for each bulk piece obtained in this way.
We color each vertex that is adjacent to a boundary region by the label of that region -- including the purifying region $O = A_{n+1} = \partial X - \bigcup_{i=1}^n A_i$.
Finally, between any two vertices we add an edge with weight the codimension-one surface area of the common boundary of the corresponding pieces' closures, divided by $4 G_N$ (unless this area is zero, in which case we may omit the edge).
See \cref{fig:tograph-2} for an illustration of this construction.
Even though we have illustrated it with two-dimensional cartoons, the prescription generalizes readily to bulk geometries of higher dimension.
\begin{center}
  \includegraphics[width=0.4\linewidth]{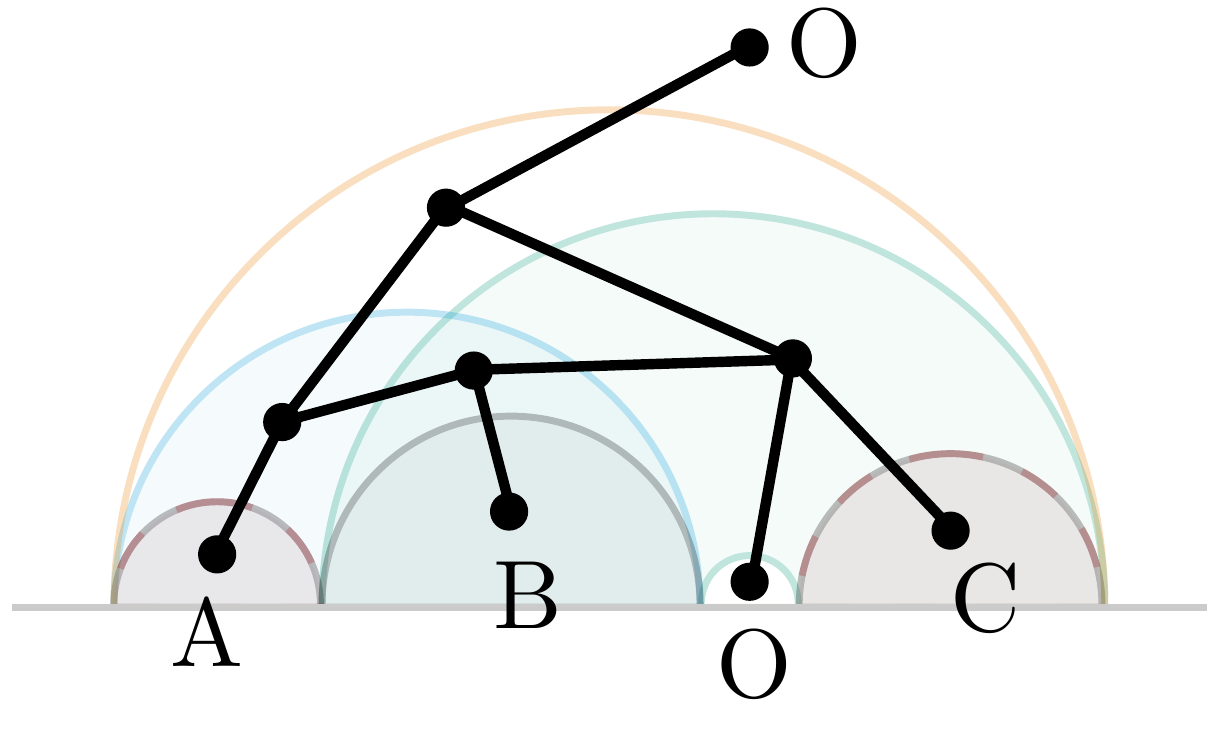}
  \captionof{figure}{The graph obtained from our construction applied to the bulk geometry and boundary regions in \cref{fig:tograph-1}.}
  \label{fig:tograph-2}
\end{center}

Before we formally describe the data thus obtained, we recall a few basic notions of graph theory:
An \emph{undirected graph} consists of a vertex set $V$ and an edge set $E \subseteq \binom V 2$.
Given \emph{edge weights} $w \colon E \rightarrow \RR$, we define the \emph{weight} of a subset of edges $F \subseteq E$ by $\abs F := \sum_{e \in F} w(e)$.
A \emph{cut} is a partition of the vertex set into two disjoint subsets, $W \cup W^c = V$.
We define $C(W) = \{ \{w,w'\} \in E : w \in W, w' \not\in W \}$ as the subset of edges that \emph{cross} the cut, i.e., with one endpoint in $W$ and the other endpoint in its complement. Thus $\abs{C(W)}$ is the sum of weights of all edges that cross the cut.
Finally, a \emph{coloring} of a subset $W \subseteq V$ of vertices by some set $C$ is simply a function $W \rightarrow C$.
We are thus led to the following definition:

\begin{dfn}
\label{dfn:graph model}
  Let $(V,E)$ be a undirected graph with non-negative edge weights $w \colon E \rightarrow \RR_{\geq 0}$.
  Let $\partial V \subseteq V$ be a subset of vertices, called the \textbf{boundary vertices}, with coloring $b \colon \partial V \rightarrow [n+1]$.
  All other vertices are called \textbf{bulk vertices}.
  This data together constitutes a \textbf{graph model}.

  \smallskip
  \noindent
  For each subset $I \subseteq [n]$, we define the \textbf{discrete entropy} by the formula
  \begin{equation}
  \label{eq:discrete entropy}
    S^*(I) = \min_{W \cup W^c = V} \abs{C(W)}
  \end{equation}
  where the optimization is over all cuts $W \cup W^c = V$ such that $W$ contains precisely those boundary vertices that are colored by $I$, i.e.,
  $W \cap \partial V = b^{-1}(I)$. We will refer to such a cut as an $I$-cut.
\end{dfn}

\Cref{dfn:graph model} can be readily extended to non-simple graphs which have loops and parallel edges.
We will see below that this does not lead to a richer spectrum of entropies.
We remark that, by the max-flow min-cut theorem from graph theory, the discrete entropy of a region $I$ can also be defined as the maximal flow from the boundary vertices colored by $I$ to those colored by $[n+1]\setminus I$ that does not exceed the capacities given by the edge weights.

Our definition of entropy in the graph model is justified by the following key lemma.
It asserts that the discrete entropy formula \cref{eq:discrete entropy} reproduces faithfully the Ryu-Takayanagi entropies \cref{eq:ryu takayanagi} in the original bulk geometry.

\begin{lem}
\label{lem:tograph}
  Given a bulk geometry $X$ and boundary regions $A_1, \dots, A_n \subseteq \partial X$, construct the associated graph model as described above.
  Then $S(I) = S^*(I)$.
\end{lem}
\begin{proof}
  Recall that the vertex set of the graph model corresponds to the finitely many pieces into which the bulk geometry is cut by minimal Ryu--Takayanagi surfaces.
  We shall denote these pieces by $a(v)$ for $v \in V$.
  With any cut $W \cup W^c = V$ we may therefore associate the bulk region $a(W) = \bigcup_{w \in W} a(w)$.
  Its boundary can be decomposed into two parts, the boundary and the bulk contribution,
  \begin{align*}
    \partial a(W)
  &= (\partial a(W) \cap \partial X) \cup (\partial a(W) \cap (X \setminus \partial X)) \\
  &= \bigcup_{w \in W \cap \partial V} (\partial a(w) \cap \partial X) \cup \bigcup_{\{w,w'\} \in C(W)} (a(w) \cap a(w')),
  \end{align*}
  where we have used that the bulk contribution arises from common boundaries between any two bulk pieces $a(w)$ and $a(w')$ such that $w \in W$ but $w' \not\in W$.
  Now if $W$ is an $I$-cut then the first contribution can be further simplified to
  \[
    \bigcup_{w \in W \cap \partial V} (\partial a(w) \cap \partial X)
  = \bigcup_{w \in b^{-1}(I)} (\partial a(w) \cap \partial X)
  = \bigcup_{i \in I} A_i = A_I.
   \]
  We conclude that the bulk surface $A'(W) = \bigcup_{\{w,w'\} \in C(W)} (a(w) \cap a(w'))$ is homologous to the boundary region $A_I$.

  We have thus obtained a bijection that associates to any $I$-cut $W \cup W^c = V$ a certain bulk surface $A'(W)$ that is homologous to $A_I$.
  Since the weight of an edge $e = \{w,w'\}$ was precisely defined in terms of the surface area of the codimension-one piece of $a(w) \cap a(w')$, it is immediate that the weight of the cut $W$ agrees with the surface area of the surface $A'(W)$, divided by $4 G_N$,
  \[ \abs{\partial W} = \frac{\abs{A'(W)}}{4 G_N}. \]
  This shows that $S^*(I) \geq S(I)$.
  But note that we can always obtain the minimal Ryu-Takayanagi surfaces $A'_I$ that were used in the construction of the bulk decomposition as some $A'(W)$ (take $W$ to be the set of all vertices $w$ such that $a(w) \subseteq a_I$, where $a_I$ is the bulk region with $\partial{a_I} = A_I \cup A'_I$).
  Thus $S(I) = S^*(I)$, as was asserted in the lemma.
\end{proof}

It is easy to write down graph models for the extreme rays \cref{eq:ray2,eq:ray3,eq:ray4} of the holographic entropy cone for $n \leq 4$ regions.
The result, which can be obtained by following the general construction outlined above, is displayed in \cref{fig:raygraphs234}.

\begin{figure}
\centering
  \includegraphics[width=0.8\linewidth]{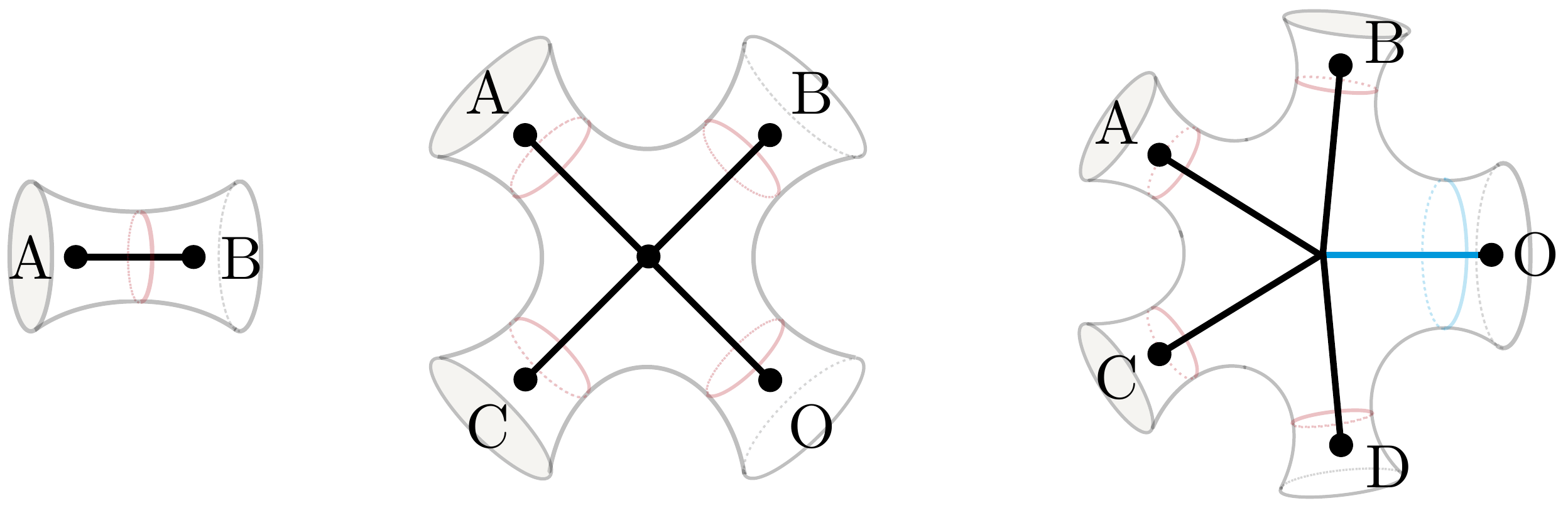}
  \captionof{figure}{Graph models of the extreme rays of the holographic entropy cone for $n \leq 4$ regions. All black edges have the same weight; the blue edge has twice that weight.}
  \label{fig:raygraphs234}
\end{figure}

In \cref{subsec:universal graph models}, we sketch an alternative way of constructing graph models based on discretizing the given geometry.
In contrast to the procedure described above, it is only approximate but does work without \emph{a priori} choice of boundary regions.

\paragraph{Graph transformations.}
It is easy to find graph transformations that preserve the discrete entropies.
We list several such transformations, which will prove useful in \cref{subsec:fromgraph} below to bring graph models into a canonical form:
\begin{itemize}
  \item Removing edges of zero weight as well as loops.
  \item Removing isolated bulk and boundary vertices, i.e., vertices that are not connected to any edge.
  \item Unifying boundary vertices (\cref{fig:graphtrafos}, (a)).
  \item Removing bulk vertices of degree one and two  (\cref{fig:graphtrafos}, (b) and (c)).
  \item Unifying parallel edges (\cref{fig:graphtrafos}, (d)).
  \item Splitting up bulk vertices of degree larger than three  (\cref{fig:graphtrafos}, (e)).
\end{itemize}
In each case, it is straightforward to verify that all discrete entropies are preserved.
We note that transformation (c) in \cref{fig:graphtrafos} can introduce parallel edges or loops.
However, loops can always be removed and parallel edges can be consolidated by using transformation (d).

\begin{center}
\includegraphics[width=0.8\linewidth]{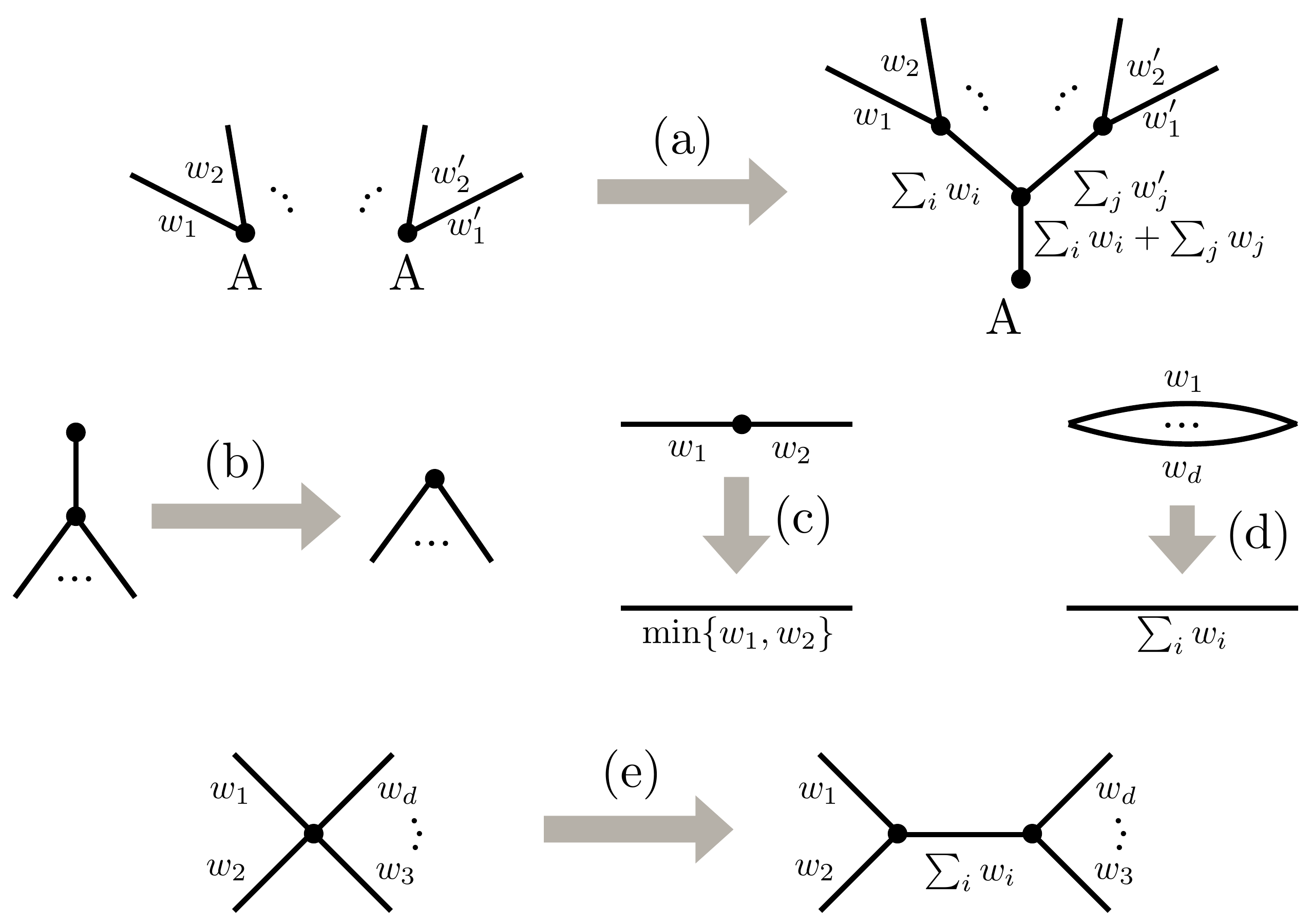}
\captionof{figure}{Entropy-preserving graph transformations: (a) Unifying boundary vertices; (b) and (c) removing bulk vertices of degree one and two; (d) unifying parallel edges; (e) splitting up bulk vertices of degree $d>3$.}
\label{fig:graphtrafos}
\end{center}

\subsection{From graphs to geometries}
\label{subsec:fromgraph}

In this section we shall establish a converse to \cref{lem:tograph}:
For any graph model we will describe how to construct a bulk geometry and boundary regions such that the Ryu-Takayanagi formula reproduces the discrete entropies.
From this we will be able to conclude that the graph model gives a completely equivalent combinatorial description of holographic entropy (see \cref{thm:holographic cone} below).
The bulk geometry that we will construct will be a two-dimensional manifold, i.e.\ a surface, and we will require the bulk geometry to have constant negative curvature since this is a sufficient condition for the geometry to be holographically dual to a CFT state for which the Ryu-Takayanagi formula is applicable (see \cref{sec:physics} below).
However, the construction can be easily generalized to higher dimension and we will comment on this in \cref{subsec:higherdim} below.

The construction goes as follows.
Given a graph model, we first subject it to a number of entropy-preserving graph transformations of the kind described in the previous section
until we obtain a graph model with the following properties:
\begin{itemize}
\item Each edge has positive weight.
\item Each bulk vertex is trivalent (i.e., connected to precisely three edges).
\item Each boundary vertex is connected to a single edge and each boundary color appears at most once.
\end{itemize}
%
%
%
For each bulk vertex, we now insert a hyperbolic `pair of pants' of constant scalar curvature $K < 0$ to be determined later; we take each geodesic boundary cycle to be of length $4 G_N w(e)$, where $w(e)$ is the weight of the corresponding edge $e$ and $G_N$ the Newton constant (in 2+1 dimensions).
For each boundary vertex, we insert a `half-collar' (of arbitrary width) whose geodesic boundary cycle likewise has length $4 G_N$ times the weight of the corresponding edge.
Finally, we glue together each pair of boundary cycles that corresponds to an edge in the graph (with arbitrary twists).

We have thus constructed a hyperbolic surface $X$ of constant negative curvature.
The connected components of its boundary $\partial X$ are in one-to-one correspondence with the set of boundary vertices in the graph model.
We shall thus define the boundary regions $A_1, \dots, A_n$ accordingly (and set $A_i=\emptyset$ if there is no corresponding boundary vertex).
\Cref{fig:fromgraph} illustrates the construction.
\begin{figure}
\centering
\includegraphics[width=0.6\linewidth]{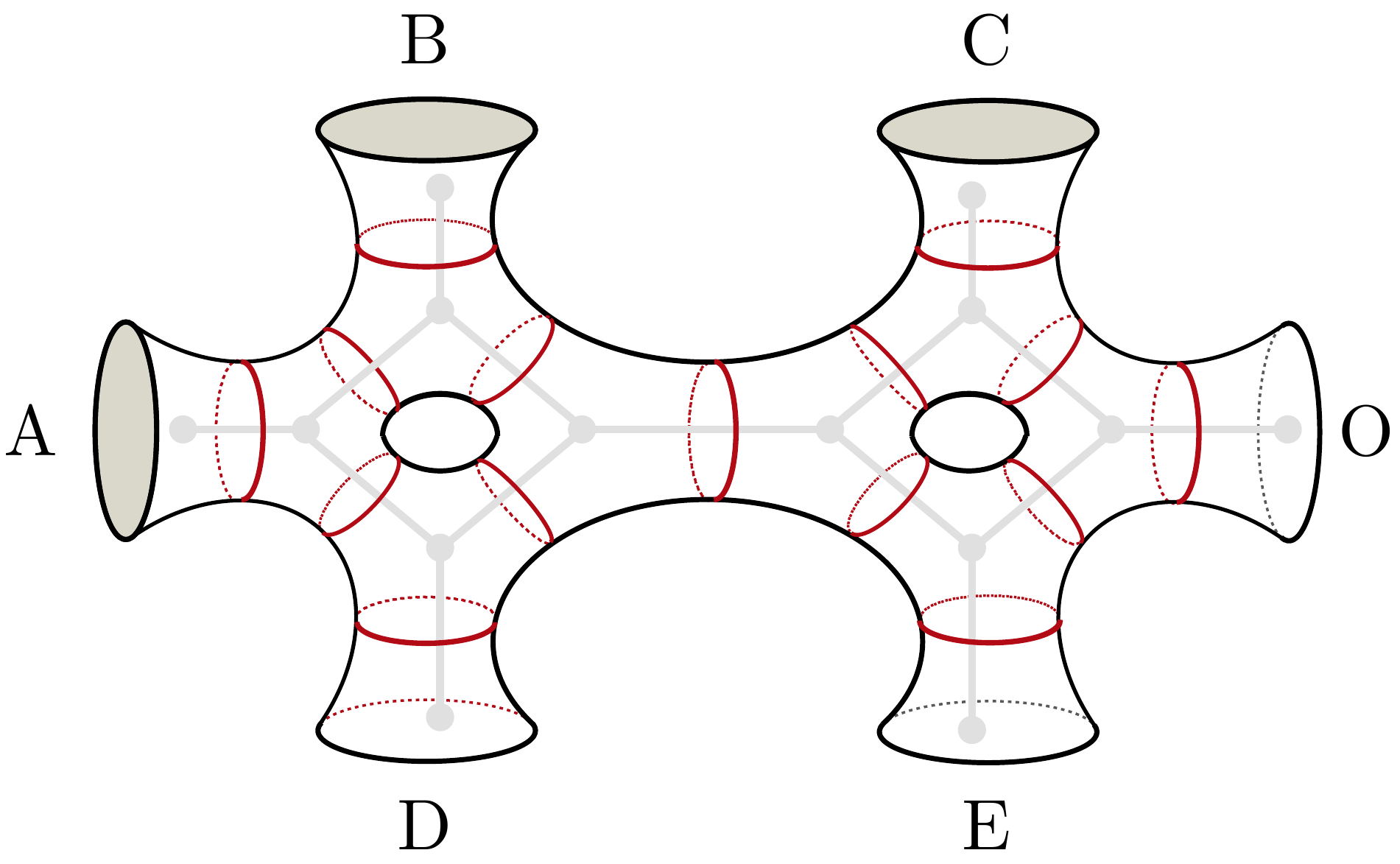}
\captionof{figure}{Construction of a hyperbolic surface from a graph model.}
\label{fig:fromgraph}
\end{figure}
Here we have drawn in red the geodesics along which we have glued the pairs of pants and half-collars.
The following lemma now provides the desired converse to \cref{lem:tograph}:

\begin{lem}
\label{lem:fromgraph}
  Given a graph model, construct an associated hyperbolic surface $X$ and boundary regions $A_1, \dots, A_n \subseteq \partial X$ as described above.
  If the scalar curvature $K < 0$ is chosen to be sufficiently negative, $S(I) = S^*(I)$.
\end{lem}
\begin{proof}
  Since each boundary region $A_I$ is a closed manifold, the corresponding minimization in the Ryu-Takayanagi formula \cref{eq:ryu takayanagi} is over closed cycles $A'$ homologous to the boundary.
  Any minimizing cycle $A'$ is necessarily given by a disjoint union of simple closed geodesics.
  The construction of $X$ gives us an ample supply of simple closed geodesics, namely those along which we had glued the pairs of pants and half-collars.
  Let us call these the \emph{gluing geodesics}.

  We will now show that for sufficiently negative scalar curvature, any minimizing cycle $A'$ is only composed of gluing geodesics.
  To see this, let $\gamma$ be a simple closed geodesic that is not a gluing geodesic.
  Since neither a hyperbolic pair of pants nor a half-collar contain any simple closed geodesics apart from their boundaries, $\gamma$ will necessarily intersect one of the gluing geodesics transversally. 
  The lengths of the latter are by construction no larger than $\ell_{\max} := 4 G_N \max_e w(e)$.
  Thus the hyperbolic collar theorem (e.g., \cite[Corollary 4.1.2]{buser}) asserts that the length $\ell_\gamma$ of $\gamma$ satisfies
  \[ \sinh \frac {\ell_\gamma} {2\sqrt{-K}} \sinh \frac {\ell_{\max}} {2\sqrt{-K}} > 1. \]
  By choosing $K$ sufficiently negative we can thus ensure that $\ell_\gamma$ is larger than $4 G_N \sum_e w(e)$, the total length of all gluing geodesics.
  It follows at once that $\gamma$ will never appear in a minimizer of \cref{eq:ryu takayanagi}, since we can always do better, e.g., by choosing $A'$ as the union of the gluing geodesics that bound the half-collars corresponding to $A_I$.

  The above discussion shows that we may restrict the minimization in the Ryu-Takayanagi formula for $S(I)$ to cycles $A'$ that are composed of gluing geodesics.
  It is not hard to see that such a cycle $A'$ is homologous to $A_I$ if and only if the corresponding set of edges in the graph model is induced by a cut $W \cup W^c = V$.
  Since by construction the length of a gluing cycle is equal to $4 G_N$ times the weight of the corresponding edge, we conclude that the Ryu-Takayanagi entropy $S(I)$ indeed agrees with the discrete entropy $S^*(I)$.
\end{proof}

While our graphs do not specify the surface uniquely due to the arbitrary choice of twists, \cref{lem:fromgraph} guarantees that the holographic entropies do not depend on these choices.
Similar reasoning based on the collar theorem also enters in the construction of certain limit points of Teichm\"uller space.
However, there the focus is on the behavior of \emph{large} (transversal) cycles, while we are interested in the \emph{smallest} cycles due to the minimization in the Ryu-Takayanagi formula.

As an immediate consequence of \cref{lem:tograph,lem:fromgraph}, the graph model gives a completely equivalent, combinatorial characterization of the holographic entropy cone. Thus we obtain the following theorem:

\begin{thm}
\label{thm:holographic cone}
  The holographic entropy cone $\calC_n$ can be equivalently defined in terms of
  (1) compact Riemannian manifolds (our original definition~\ref{obs:convex cone}) and 
  (2) graph models.
  Moreover, (1) can be restricted to two-dimensional manifolds of constant negative curvature.
\end{thm}

We remark that any Riemannian manifold with constant scalar curvature and with asymptotically hyperbolic boundaries can be understood as the time slice of a multiboundary wormhole geometry (see \cref{sec:physics} below).

\subsection{Higher dimensions}
\label{subsec:higherdim}

Though the two-dimensional construction in \cref{subsec:fromgraph} above is sufficient to prove the existence of the bulk geometry for each graph model,
we note that a similar construction exists in higher dimensions.
In fact, it is easier to construct examples in higher dimensions since the requirement of constant scalar curvature is less restrictive;
in two dimensions, this requirement demands that the surface is locally hyperbolic and indeed completely determines its local geometry; this is not the case in higher dimensions.

To give an explicit example, take the Riemann surface $\Sigma$ constructed in \cref{lem:fromgraph} and tensor it with the $(d-2)$-dimensional sphere $S^{d-2}$ of scalar curvature $K'$.
The total space $\Sigma \times S^{d-2}$ is a $d$-dimensional manifold with constant scalar curvature $K + K'$. Let us choose $K'$ so that $K+ K' < 0$.
The space has boundaries with the topology of $S^1 \times S^{d-2}$.
These boundaries are not asymptotically hyperbolic, but we can repair them as follows.
Near each boundary of $\Sigma \times S^{d-2}$, we can choose the metric as
\begin{equation*}
  ds^2 = \frac{1}{|K|} \left( \frac{dr^2}{1+r^2} + (1+r^2) d\tau^2 \right) + \frac{1}{K'} d\Omega^2_{d-2},
\end{equation*}
where $\tau$ is periodic as $\tau \sim \tau + 2\pi$, representing the $S^1$ direction on the boundary, and where $d\Omega_{d-2}^2$ is the metric of the $(d-2)$-dimensional unit sphere.
On the other hand, the metric on the $d$-dimensional hyperbolic space can be chosen as
\begin{equation*}
  ds^2 = \frac{1}{|K+ K'|} \left( \frac{dr^2}{1+r^2} + (1+r^2) d\tau^2+ r^2 d\Omega^2_{d-2}\right),
\end{equation*}
where we assumed that the space has the same scalar curvature, $K' + K$. These two geometries can be interpolated by
\begin{equation}
  ds^2 = \frac{f(r)}{|K|} \left( \frac{dr^2}{1+r^2} + (1+r^2) d\tau^2 \right) + \frac{g(r)}{K'} d\Omega^2_{d-2},
\label{eq:ansatz}
\end{equation}
where $f(r) = g(r) = 1$ for $r < r_0$ for some $r_0$, and $f(r ) \sim |K|/|K+K'|$ and $g(r) \sim r^2 K'/|K+K'|$ for $r \rightarrow \infty$.
The constant scalar curvature condition for this metric is a second-order nonlinear differential equation on $f(r)$ and $g(r)$.
Since we can choose these two functions arbitrarily as far as they are positive and obey the boundary conditions at $r = r_0$ and $r \rightarrow \infty$, we can always find a solution for the single differential equation.  In fact, there are infinitely many solutions. The transitional region \cref{eq:ansatz} does not add new minimal surfaces homologous to the boundary and the holographic entropy vector associated with the $d$-dimensional bulk geometry $\Sigma \times S^{d-2}$ remains the same as the one for $\Sigma$ (up to overall rescaling) provided that we can choose
\begin{equation}
r_0 > \left( \frac{|K + K'|^{d-1}}{|K| K'^{d-2}} \right)^\frac{1}{2(d-2)}, \qquad \partial_r \left[ \left( 1 + r^2 \right) f g^{d-2} \right] > 0.
\end{equation}
Ansatz \cref{eq:ansatz} is by no means unique. The purpose of this example is to illustrate how it is easy to find a higher-dimensional geometry corresponding to each graph model.
It also shows that the holographic entropy cones do not change when we restrict to fixed bulk dimension $d > 2$, constant negative curvature, and asymptotically hyperbolic boundaries (generalizing \cref{thm:holographic cone}).

\subsection{Polyhedrality}
\label{subsec:properties}

We will now use the characterization of the holographic entropy cone in terms of graph models to gain further structural insight non-obvious from its original definition. We start with a basic lemma:

\begin{lem}
\label{lem:fixed graph model}
  Any entropy vector in $\calC_n$ can be explained by a complete graph on $2^{2^n-1}$ vertices and fixed boundary coloring.
\end{lem}
\begin{proof}
  Let $\mathcal I_n = \{ \emptyset \neq I \subseteq [n] \}$ be the set of non-empty subsets of $[n]$.
  We define the vertex set of our complete graph to be $V := \{0,1\}^{\mathcal I_n}$, the set of bitstrings indexed by $\mathcal I_n$. Clearly, $V$ is of cardinality $2^{2^n-1}$.
  For each $i \in [n+1]$, we define a bitstring $x_i \in V$ by $(x_i)_I = 1$ if and only if $i \in I$ (note that $x_{n+1}$ is always the bitstring that is all zeros).
  These bitstrings are our boundary vertices: we color each $x_i$ by $i$.

  We now describe the transformation of an arbitrary graph model into one that uses the universal vertex set and boundary coloring.
  We will proceed by what is in essence an algebraic version of the construction given in \cref{subsec:tograph}.
  To start, we choose for each boundary region $I \in \mathcal I_n$ an $I$-cut $W_I$ that realizes the discrete entropy $S^*(I)$ in the original graph.
  For each bitstring $x \in V$, we define $W(x) := \bigcap_{I \in \mathcal I_n} W_I^{x_I}$, where we denote $W^1_I := W_I$ and $W^0_I := W_I^c$.
  The regions $W(x)$ partition the vertex set of the original graph into $2^{2^n-1}$ disjoint subsets (of which some may be empty).
  We note that each $W(x_i)$ contains all boundary vertices of the original graph that are colored by $i$.
  Finally, let $E(x,y)$ denote the set of edges in the original graph with one endpoint in $W(x)$ and one endpoint in $W(y)$.
  We define the weight of the edge between bitstrings $x$ and $y$ to be $w(x, y) = \sum_{e \in E(x,y)} w(e)$, the total weight of all edges between $W(x)$ and $W(y)$.
  Note that $w(x,y) = 0$ if $E(x,y) = \emptyset$.
  Given these definitions, it is not hard to verify that all discrete entropies $S^*(I)$ are preserved in the graph thus obtained.
\end{proof}


\Cref{thm:holographic cone,lem:fixed graph model} imply that we may characterize each holographic entropy vector in $\calC_n$ by a graph model on a fixed graph $(V,E)$ with fixed boundary coloring. Thus the only varying data are the edge weights $w \colon E \rightarrow \RR_{\geq0}$. It follows that the holographic entropy cone is given by
\begin{equation*}
  \calC_n = \{ s^*(w) \,:\, w \colon E \rightarrow \RR_{\geq0} \}
\end{equation*}
where we write $s^*(w) = (S^*_w(I))_{\emptyset \neq I \subseteq [n]}$ for the discrete entropy vector corresponding to the choice of edge weights $w$ in the otherwise fixed graph model.

\begin{prp}
\label{prp:polyhedral}
  The holographic entropy cone $\calC_n$ is a rational polyhedral cone.
\end{prp}
\begin{proof}
  According to the preceding discussion, we may work with a fixed graph and boundary coloring.
  This also fixes the $I$-cuts for each boundary region $I$; let us denote them by $W_I^{(1)}, \dots, W_I^{(n_I)}$.
  Then the discrete entropy \cref{eq:discrete entropy} is given by the minimization
  \begin{equation}
  \label{eq:finite minimization}
  S^*_w(I) = \min \{ \abs{C(W_I^{(1)})}, \dots, \abs{C(W_I^{(n_I)})} \}.
  \end{equation}
  where each $\abs{C(W_I^{(k)})} = \sum_{e \in C(W_I^{(k)})} w(e)$ is a linear function of the edge weights $w$, which we may think of as elements of the orthant $\RR^E_{\geq 0}$.
  Now consider the linear hyperplanes
  \[ \mathcal H_I^{(k,l)} = \{ w \in \RR^E_{\geq0} : \abs{C(W_I^{(k)})} = \abs{C(W_I^{(l)})} \} \]
  for all $I$ and $1 \leq k \neq l \leq n_I$.
  They partition the orthant of edge weights into finitely many rational polyhedral subcones,
  $\RR^E_{\geq 0} = \bigsqcup_{\mathcal W} \mathcal W$, each of which is defined by a maximal number of inequalities of the form
  $\abs{C(W_I^{(k)})} \geq \abs{C(W_I^{(l)})}$.
  On each such subcone $\mathcal W$, the discrete entropy \cref{eq:finite minimization} of any boundary region $I$ -- and therefore the entropy vector $s(w)$ itself -- is given by a linear function of $w$ with integer coefficients.
  As the image of a rational polyhedral cone under such a linear map is again a rational polyhedral cone, we conclude that the holographic entropy cone is a union of finitely many rational polyhedral cones, $\calC_n = \bigcup_{\mathcal W} s(\mathcal W)$, and therefore itself a rational polyhedral cone (since we already know that $\calC_n$ is a convex cone).
\end{proof}

Remarkably, \cref{prp:polyhedral} implies that, for each fixed number of regions $n$, there is a \emph{finite} number of independent linear entropy inequalities, corresponding to the facets of the holographic entropy cone $\calC_n$.
Equivalently, each holographic entropy cone is spanned by a finite number of extreme rays, which we have seen can be represented, e.g., by hyperbolic surfaces.
Such a hyperbolic surface can in turn be understood as the time slice of a multiboundary wormhole geometry \cite{skenderis_van_rees_2011,Brill}.
This is in agreement with and generalizes our findings for $n \leq 4$ regions in \cref{subsec:few regions}.

The holographic situation is also in stark contrast to the Shannon and von Neumann entropy.
For these, not only are there extreme rays for $n \geq 3$ that can only be attained approximately, so that the corresponding entropy cones are not closed, but there are in fact linear inequalities that constrain some of the lower-dimensional faces when $n \geq 4$ \cite{ZhangYeung97,LindenWinter05,CadneyLindenWinter12}.
It is moreover known for the Shannon entropy (and likewise conjectured for the von Neumann entropy) that the cones are not polyhedral for $n \geq 4$ \cite{Matus07}.


\section{New constraints on holographic entropies}
\label{sec:inequalities}

In this section, we describe a new combinatorial method for establishing holographic entropy inequalities.
As we shall explain below, it is based on exhibiting a certain contraction map of Hamming cubes as a certificate for the correctness of a given entropy inequality.
We then use this method of \emph{proof by contraction} to establish an infinite family of hitherto unknown holographic entropy inequalities, and we comment on the novel features of holographic entropy for five or more regions.
We also give a greedy algorithm for finding proofs by contraction for a given entropy inequality.

\medskip

Before we describe the method, we first recall the holographic proof of strong subadditivity from \cite{headrick_takayanagi_2007}:
Let $AB'$ and $BC'$ denote the minimal Ryu-Takayanagi hypersurfaces for $S(AB)$ and $S(BC)$, respectively, and $ab$ and $bc$ the corresponding bulk regions.
That is, $\partial ab = AB \cup AB'$ and $\partial bc = BC \cup BC'$.
Now define two new bulk regions by $b := ab \cap bc$ and $abc := ab \cup bc$.
Clearly, their boundaries decompose as $\partial b = B \cup B'$ and $\partial abc = ABC \cup ABC'$, where $ABC'$ and $B'$ denote the respective bulk parts.
One now argues that $B'$ and $ABC'$ can be assembled from $AB'$ and $BC'$ by cutting and pasting (see \cref{fig:classical ssa proof} for an illustration).
This then shows the first inequality in
\[ S(AB) + S(BC) = \frac1{4G_N} \left( \abs{AB'} + \abs{BC'} \right) \geq \frac1{4G_N} \left( \abs{B'} + \abs{ABC'} \right) \geq S(B) + S(ABC), \]
since, in general, $AB'$ and $BC'$ may not be completely exhausted by $B'$ and $ABC'$ (unlike suggested in \cref{fig:classical ssa proof}).
The second inequality is due to the minimization in the Ryu-Takayanagi formula, since the bulk surfaces constructed will in general not be minimal in their homology class.
This concludes the holographic proof of strong subadditivity \cref{eq:ssa}.
The proof of the monogamy inequality \cref{eq:mmi} in \cite{hayden_headrick_maloney_2013} proceeds similarly by reassembling bulk surfaces and regions.

\begin{figure}
\centering
\includegraphics[width=0.8\linewidth]{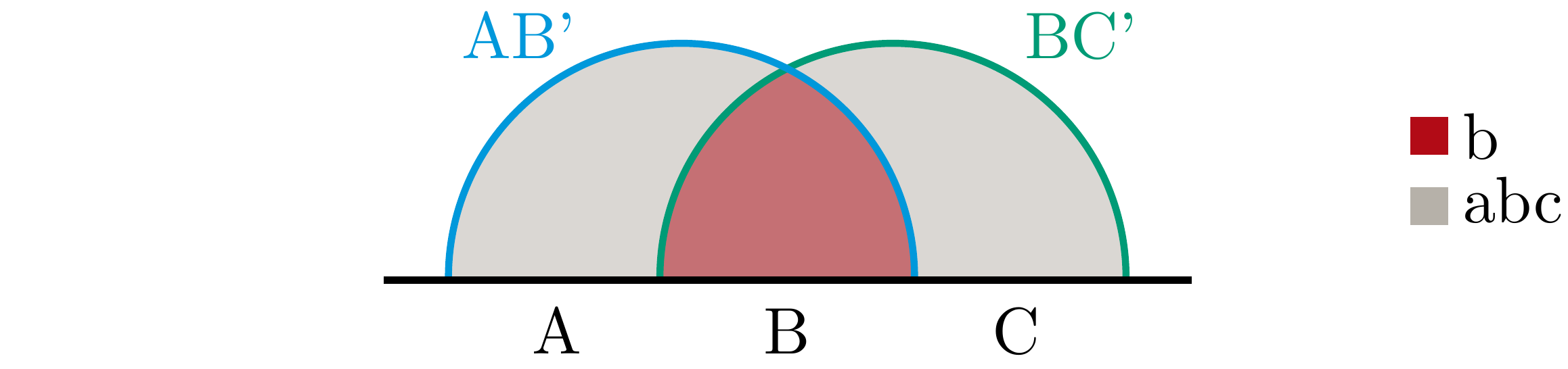}
\caption{Illustration of the holographic proof of strong subadditivity from \cite{headrick_takayanagi_2007}.}
\label{fig:classical ssa proof}
\end{figure}

\subsection{Proofs by contraction}
\label{subsec:contractions}

We now show that the proofs in \cite{headrick_takayanagi_2007,hayden_headrick_maloney_2013} are instances of a general combinatorial method.
For this, we consider a general entropy inequality on $n$ parties given in the form
\begin{equation}
\label{eq:entropy inequality}
  \sum_{l=1}^L \alpha_l S(I_l) \geq \sum_{r=1}^R \beta_r S(J_r),
\end{equation}
where the $\alpha_1, \dots, \beta_R > 0$ are positive coefficients and $I_1, \dots, J_R \subseteq \{1,\dots,n\}$ the corresponding subsystems.
We can conveniently encode the latter in terms of the following \emph{occurrence vectors}:
\begin{equation}
\label{eq:occurrence vectors}
\begin{aligned}
  x_i &:= (i \in I_l)_{l=1}^L \in \{0,1\}^L, \\
  y_i &:= (i \in J_r)_{r=1}^R \in \{0,1\}^R.
\end{aligned}
\end{equation}
For the purifying region, we accordingly define $x_{n+1} \equiv 0$ and $y_{n+1} \equiv 0$ to be zero bitstrings.
Finally, we define the \emph{weighted Hamming norms} $\norm{v}_\alpha := \sum_{l=1}^L \alpha_l \abs{v_l}$.

\begin{thm}[`Proof by contraction']
\label{thm:proof by contraction}
  Let $f \colon \{0,1\}^L \rightarrow \{0,1\}^R$ be a $\norm{\cdot}_\alpha$-$\norm{\cdot}_\beta$-contraction, i.e.,
  \begin{equation}
  \label{eq:contraction}
    \norm{f(x) - f(x')}_\beta \leq \norm{x-x'}_\alpha \qquad (\forall x,x' \in \{0,1\}^L).
  \end{equation}
  If $f(x_i) = y_i$ for all $i=1,\dots,n+1$ then \cref{eq:entropy inequality} is a valid entropy inequality.
\end{thm}
\begin{proof}
  By \cref{thm:holographic cone}, it suffices to show that the entropy inequality holds for the discrete entropy of an arbitrary graph model.
  %
  We first choose a minimal cut $W_l$ for each of the regions $I_l$ appearing in the left-hand side of \cref{eq:entropy inequality}. That is, $W_l$ is an $I_l$-cut and $S^*(I_l) = \abs{C(W_l)}$.
  For each bitstring $x \in \{0,1\}^L$, we now define a subset $W(x) := \bigcap_{l=1}^L W_l^{x_l}$ by inclusion/exclusion, where we denote $W_I^1 := W_I$ and $W_I^0 := W_I^c$.
  Then the $W(x)$ form a partition of the vertex set of the graph. 
  Each $W(x_i)$ contains all boundary vertices colored by $i$,
  and we have that
  \begin{equation}
  \label{eq:original cuts}
    W_l  = \bigcup \{ W(x) : x \in \{0,1\}^L \text{ with } x_l = 1 \}.
  \end{equation}
  %
  We now define a cut for each of the regions $J_r$ that appear on the right-hand side of \cref{eq:entropy inequality} by reassembling the pieces $W(x)$ according to the function $f$:
  \begin{equation}
  \label{eq:reassembled cuts}
    U_r := \bigcup \{ W(x) : x \in \{0,1\}^L \text{ with } f(x)_r = 1 \}
  \end{equation}
  Since we assume that $f(x_i) = y_i$ for all $i=1,\dots,n+1$, each $U_r$ is a $J_r$-cut. Indeed, this follows from
  \[
    W(x_i) \subseteq U_r
    \Leftrightarrow
    f(x_i)_r = 1
    \Leftrightarrow
    (y_i)_r = 1
    \Leftrightarrow
    i \in J_r.
  \]
  %
  The crucial step now is to establish the following inequality, which in fact holds for \emph{arbitrary} cuts $W_l$:
  \begin{equation}
  \label{eq:bulk inequality in proof}
    \sum_{l=1}^L \alpha_l \abs{C(W_l)} \geq \sum_{r=1}^R \beta_r \abs{C(U_r)}
  \end{equation}
  To see that \cref{eq:bulk inequality in proof} is correct, it is useful to introduce for any pair of bitstrings $x, x' \in \{0,1\}^L$ the set $E(x,x')$ of edges between $W(x)$ and $W(x')$.
  The $E(x,x')$ form a partition of the edge set of the graph model under consideration.
  By \cref{eq:original cuts}, the edges in $E(x,x')$ cross $W_l$ if and only if $x_l \neq x'_l$, 
  so that $C(W_l) = \bigcup_{x_l \neq x'_l} E(x,x')$.
  Therefore,
  \begin{align*}
      &\sum_{l=1}^L \alpha_l \abs{C(W_l)}
    = \sum_{l=1}^L \alpha_l \sum_{\{x,x'\} : x_l \neq x'_l} \abs{E(x,x')}
    = \sum_{l=1}^L \alpha_l \sum_{\{x,x'\}} \abs{x_l - x'_l} \, \abs{E(x,x')} \\
    = &\sum_{\{x,x'\}} \abs{E(x,x')} \sum_{l=1}^L \alpha_l \abs{x_l - x'_l}
    = \sum_{\{x,x'\}} \abs{E(x,x')} \, \norm{x - x'}_\alpha.
  \end{align*}
  Likewise, \cref{eq:reassembled cuts} implies that
  \begin{equation*}
    \sum_{r=1}^R \beta_r \abs{C(U_r)} = \sum_{\{x,x'\}} \abs{E(x,x')} \, \norm{f(x) - f(x')}_\beta.
  \end{equation*}
  Thus the inequality \cref{eq:bulk inequality in proof} is a consequence of the contraction property \cref{eq:contraction}.
  %
  It is now straightforward to conclude the proof of the entropy inequality \cref{eq:entropy inequality}:
  \begin{equation}
  \label{eq:final argument}
      \sum_{l=1}^L \alpha_l S^*(I_l)
    = \sum_{l=1}^L \alpha_l \abs{C(W_l)}
    \geq \sum_{r=1}^R \beta_r \abs{C(U_r)}
    \geq \sum_{r=1}^R \beta_r S^*(J_r),
  \end{equation}
  where the first inequality is \cref{eq:bulk inequality in proof}; the second inequality comes from the fact that each $U_r$ is a cut for the boundary region $J_r$ but not necessarily a minimal cut.
\end{proof}

In view of \cref{subsec:tograph} we may directly identify any cut $W$ that appears in the proof with some bulk region $w$, the set of edges $C(W)$ that cross a cut with the bulk part $W'$ of the boundary $\partial w$, and the total weight $\lvert C(w) \rvert$ of the edges with the surface area $\lvert W' \rvert$.
It is not hard use this dictionary to translate the above proof of \cref{thm:proof by contraction} into a proof that does not explicitly rely on the graph model and \cref{thm:holographic cone}.

\paragraph{Bulk inequalities.}
As is apparent from the proof and in particular from \cref{eq:reassembled cuts}, the map $f \colon \{0,1\}^L \rightarrow \{0,1\}^R$ provides an efficient way of encoding the way the cuts (or bulk regions) are reassembled.
The contraction property \cref{eq:contraction} implies that the total weight of the right-hand side cuts is no larger than that of the left-hand side cuts (weighted appropriately with the coefficients $\alpha$ and $\beta$).
Thus the resulting bulk inequality \cref{eq:bulk inequality in proof} holds for arbitrary cuts $W_1, \dots, W_L$ and not only for minimal cuts that achieve the left-hand side discrete entropies. We record this general fact:
\begin{quote}
  Let $f \colon \{0,1\}^L \rightarrow \{0,1\}^R$ be a $\norm{\cdot}_\alpha$-$\norm{\cdot}_\beta$-contraction.
  Then the \emph{bulk inequality}
  \begin{equation}
  \label{eq:bulk inequality quote}
    \sum_{l=1}^L \alpha_l \abs{C(W_l)} \geq \sum_{r=1}^R \beta_r \abs{C(U_r)}
  \end{equation}
  holds for arbitrary cuts $W_1, \dots, W_L$ in any graph, where $U_1, \dots, U_R$ are defined according to \cref{eq:reassembled cuts}. We remark that $c(W) := \abs{C(W)}$ is known as the \emph{cut function} in graph theory.
\end{quote}
Thus the role of the initial conditions $f(x_i) = y_i$ in \cref{thm:proof by contraction} is only to ensure that the reassembled cuts $U_r$ are cuts for the appropriate boundary subsystems. This allows the bulk inequality to be lifted to a holographic entropy inequality.



\paragraph{Contractions and the hypercube.}
The contracting condition in \cref{thm:proof by contraction} can also be understood in graph-theoretical terms.
For simplicity, let us assume that all coefficients $\alpha_1 = \dots = \beta_R = 1$ (otherwise, appropriate weights have to be inserted into the following discussion).
Then the set of bitstrings $\{0,1\}^L$ can be interpreted as the set of vertices of the \emph{hypercube graph} or \emph{Hamming cube}, which we shall denote by $Q_L$.
The Hamming distance agrees with the graph distance and it is not hard to see that it suffices to check \cref{eq:contraction} for the edges of the hypercube only, i.e., only for those bitstrings $x$ and $x'$ for which $x_l \neq x'_l$ for a single component $l \in [L]$.
Let us write $x \sim x'$ if $\{x,x'\}$ form an edge of the hypercube.
Then \cref{thm:proof by contraction} admits the following graph-theoretical reformulation:
\begin{quote}
  Let $\varphi \colon Q_L \rightarrow Q_R$ be a \emph{weak graph homomorphism}, i.e., a map between the vertex sets of the hypercube graphs such that
  \[ x \sim x' \Rightarrow \varphi(x) \sim \varphi(x') \text{ or } \varphi(x) = \varphi(x'). \]
  If $\varphi(x_i) = y_i$ for all $i=1,\dots,n+1$ then \cref{eq:entropy inequality} is a valid entropy inequality.
\end{quote}
A weak graph homomorphism $Q_L \rightarrow Q_R$ is the same as an ordinary graph homomorphism into the graph $Q'_R$ obtained by adding loops at each vertex of $Q_R$. Such a graph homomorphism is also known as a \emph{coloring} of $Q_L$ by $Q'_R$.
Thus to prove a holographic entropy inequality using our method it suffices to argue that the partial coloring defined by the occurrence vectors can be extended to a full coloring of $Q_L$ by $Q'_R$.


The contraction $f \colon \{0,1\}^L \rightarrow \{0,1\}^R$ can also be represented geometrically by the subsets
\begin{equation}
\label{eq:hypercube cuts}
  C_r = \{ x \in \{0,1\}^L : f(x)_r = 1 \}
\end{equation}
of the hypercube, each of which encodes a component of the function $f$.
The conditions $f(x_i) = y_i$ then amount to requiring that each $C_r$ is a cut that contains precisely those occurrence vectors $x_i$ with $i \in J_r$, and the contraction property states that all edges of the hypercube cross at most one of the cuts $C_r$.
We summarize this last reformulation of \cref{thm:proof by contraction}:
\begin{quote}
  Let $C_1, \dots, C_R$ denote subsets of the hypercube graph $Q_L = \{0,1\}^L$ such that
  (1) each hypercube edge crosses at most one of the cuts $C_1, \dots, C_r$, and
  (2) each $C_r$ contains precisely those occurrence vectors $x_i$ that correspond to regions that appear in the $r$-th right-hand side term of the entropy inequality (i.e., $i \in J_r$).
  Then \cref{eq:entropy inequality} is a valid entropy inequality.
\end{quote}

\paragraph{Basic Examples.}
We shall now illustrate our method by giving succinct `proofs by contraction' of all hitherto known holographic entropy inequalities.
For strong subadditivity, $S(AB) + S(BC) \geq S(B) + S(ABC)$, the occurrence vectors are given in \cref{tab:ssa-contraction}:
\begin{center}
\begin{tabular}{c cc cc}
  \toprule
  & \multicolumn{2}{c}{$x$} & \multicolumn{2}{c}{$y = f(x)$} \\
  \cmidrule(lr){2-3} \cmidrule(lr){4-5}
  & AB & BC & B & ABC \\
  O & 0 & 0 & 0 & 0 \\
  C & 0 & 1 & 0 & 1 \\
  A & 1 & 0 & 0 & 1 \\
  B & 1 & 1 & 1 & 1 \\
  \bottomrule
\end{tabular}
\captionof{table}{Proof by contraction of strong subadditivity.}
\label{tab:ssa-contraction}
\end{center}
Here, we have denoted the purifying region by $O$; each row of the table lists a boundary region $i$ and the corresponding occurrence vectors $x_i$ and $y_i$.
In this case, the function $f \colon \{0,1\}^2 \rightarrow \{0,1\}^2$ is already fully defined by the condition that the occurrence vectors are mapped onto each other, and it is not hard to verify that it is a contraction.
Thus \cref{thm:proof by contraction} implies at once that holographic entropies are strongly subadditive.

Let us discuss this `proof by contraction' in some more detail.
Given minimal cuts $W_{AB}$ and $W_{BC}$ for the left-hand side of the strong subadditivity inequality, the proof proceeds by constructing the following cuts in \cref{eq:reassembled cuts}:
\begin{align*}
U_B &= W(11) = W_{AB} \cap W_{BC} \\
U_{ABC} &= W(01) \cup W(10) \cup W(11) = W_{AB} \cup W_{BC}
\end{align*}
Thus we recover precisely the same construction as in the holographic proof of strong subadditivity \cite{headrick_takayanagi_2007} that we sketched at the beginning of this section (cf.~\cref{fig:classical ssa proof}).
This is no surprise, in fact the very same construction is well-known in graph theory, where it is used to establish submodularity of the cut function.

For completeness, we also describe the corresponding hypercube picture (\cref{fig:hypercube-ssa}):
The vertices describe the two-dimensional hypercube corresponding to the left-hand side of the inequality.
Each vertex is labeled in red by a boundary region according to the occurrence vectors $x_A, \dots, x_O$.
The hypercube cuts \cref{eq:hypercube cuts} amount to $C_B = \{11\}$ and $C_{ABC} = \{01,10,11\}$ and they have been indicated in blue in the above figure.
It is immediately apparent that each hypercube edge crosses precisely one of the cuts.

\begin{center}
  \includegraphics[height=4cm]{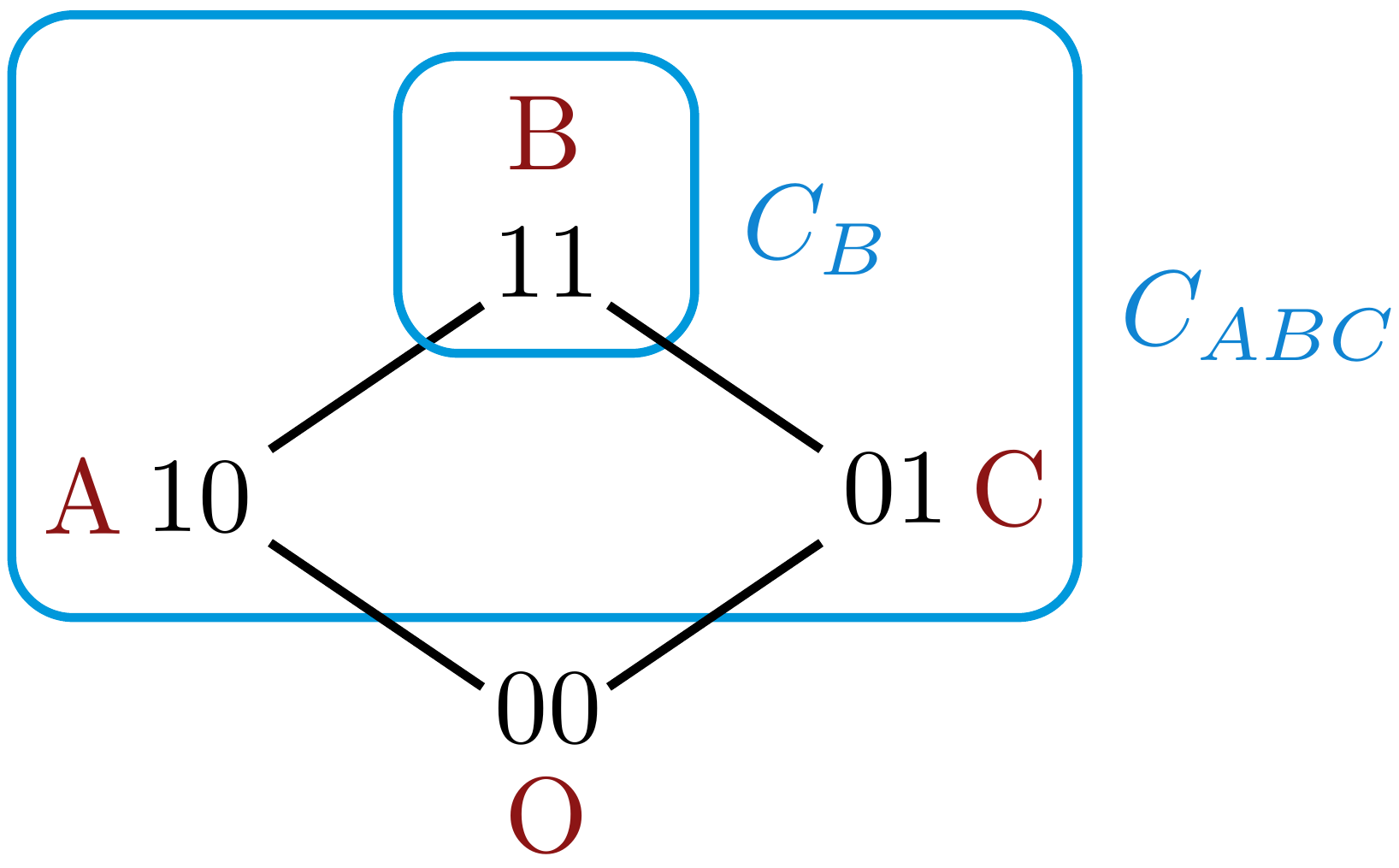}
  \captionof{figure}{Strong subadditivity of holographic entropy in the hypercube picture.}
  \label{fig:hypercube-ssa}
\end{center}

\medskip

We now sketch the corresponding construction for the monogamy of the mutual information, $S(AB) + S(BC) + S(AC) \geq S(A) + S(B) + S(C) + S(ABC)$.
In this case, there are four pairs of occurrence vectors.
\Cref{tab:mmi-contraction} lists the unique extension of this initial data to a contraction $f \colon \{0,1\}^3 \rightarrow \{0,1\}^4$.
\begin{table}
\centering
\begin{tabular}{c ccc cccc}
  \toprule
  & \multicolumn{3}{c}{$x$} & \multicolumn{4}{c}{$y = f(x)$} \\
  \cmidrule(lr){2-4} \cmidrule(lr){5-8}
  & AB & BC & AC & A & B & C & ABC \\
  O & 0 & 0 & 0 & 0 & 0 & 0 & 0 \\
    & 0 & 0 & 1 & 0 & 0 & 0 & 1 \\
    & 0 & 1 & 0 & 0 & 0 & 0 & 1 \\
  C & 0 & 1 & 1 & 0 & 0 & 1 & 1 \\
    & 1 & 0 & 0 & 0 & 0 & 0 & 1 \\
  A & 1 & 0 & 1 & 1 & 0 & 0 & 1 \\
  B & 1 & 1 & 0 & 0 & 1 & 0 & 1 \\
    & 1 & 1 & 1 & 0 & 0 & 0 & 1 \\
  \bottomrule
\end{tabular}
\captionof{table}{Proof by contraction of monogamy of the holographic mutual information.}
\label{tab:mmi-contraction}
\end{table}
This constitutes a `proof by contraction' of the monogamy inequality \cref{eq:mmi}.
The corresponding cuts are
\begin{align*}
  U_A &= W(101) = W_{AB} \cap W_{BC}^c \cap W_{AC} \\
  U_B &= W(110) = W_{AB} \cap W_{BC} \cap W_{AC}^c \\
  U_C &= W(011) = W_{AB}^c \cap W_{BC} \cap W_{AC} \\
  U_{ABC} &= \bigcup_{x \neq 000} W(x) = W_{AB} \cup W_{BC} \cup W_{AC}
\end{align*}
which is precisely the construction in \cite{hayden_headrick_maloney_2013}.
The corresponding hypercube picture is displayed below in \cref{fig:hypercube-mmi}. Again, we observe that all edges are cut precisely once.
\begin{center}
  \includegraphics[height=4cm]{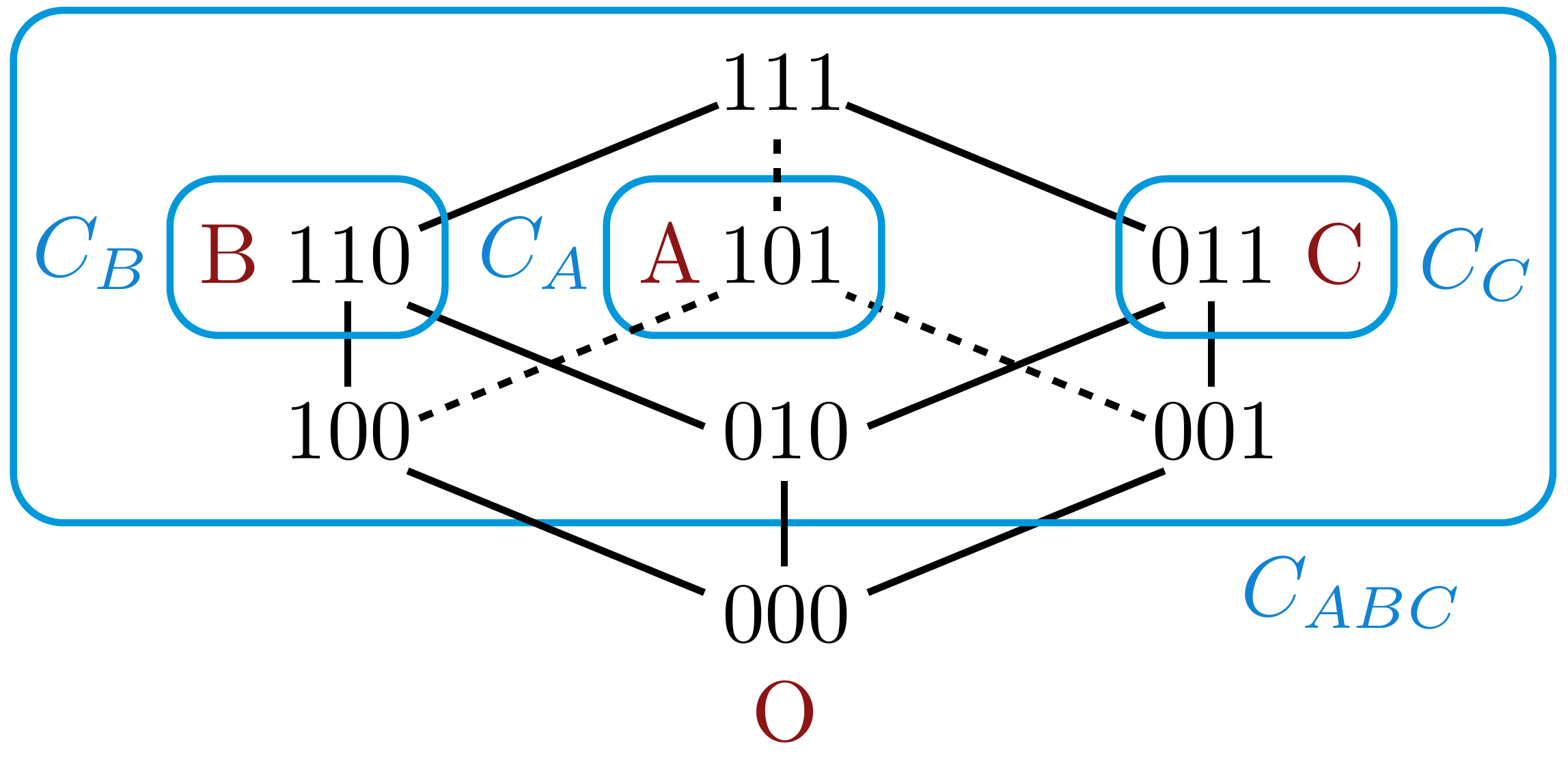}
  \captionof{figure}{Monogamy of the holographic mutual information in the hypercube picture.}
  \label{fig:hypercube-mmi}
\end{center}

\paragraph{Equality conditions.}
The proof of \cref{thm:proof by contraction} also tells us about the conditions under which a given entropy inequality holds with equality.
Indeed, it is clear from \cref{eq:final argument} that \cref{eq:entropy inequality} can only hold with equality if both (1) the cuts \cref{eq:reassembled cuts} are crossed by all edges that cross the original cuts (including weights and multiplicities) and (2) these cuts are in fact minimal.
This gives interesting geometrical information about the graph model and the underlying bulk geometry.
For example, in the case of strong subadditivity the first equality condition asserts there are no edges between $W_{AB} \setminus W_{BC}$ and $W_{BC} \setminus W_{AC}$, which is well-known from graph theory.
It would be interesting to investigate equality conditions for other holographic entropy inequalities in more detail.

\paragraph{A greedy algorithm for finding proofs by contraction.}

\begin{algorithm}
\centering
\begin{itemize}
  \item Construct the occurrence vectors $x_i$ and $y_i$ according to \cref{eq:occurrence vectors} for all $i \in [n+1]$.
  \item Start with the partially defined function $f$ that sends $f(x_i) = y_i$ for all $i \in [n+1]$.
  \item Verify that $f$ is a contraction on its domain.
    If this is not the case then the entropy inequality \cref{eq:entropy inequality} is violated for a Bell pair.
  \item Successively try to extend $f$:
    For each $x \in \{0,1\}^L$ for which the function is not yet defined, we know that any choice of $f(x)$ has to satisfy the conditions
    \[ \norm{f(x) - f(x')}_\beta \leq \norm{x-x'}_\alpha \]
    for all $x'$ for which $f(x')$ is already defined.
    If this system of constraints has no solution $f(x)$ for some $x$, fail.
    If there is a \emph{unique} solution $f(x)$ for some $x$, extend the function accordingly.
    Otherwise, pick an arbitrary $f(x)$ that satisfies the constraints.
  \item Repeat this process until the function is fully defined on the entire unit cube $\{0,1\}^L$.
\end{itemize}
\caption{Greedy algorithm for finding a `proof by contraction'.}
\label{alg:greedy}
\end{algorithm}

\Cref{thm:proof by contraction} can be readily turned into a greedy algorithm for establishing a given entropy inequality \cref{eq:entropy inequality}. 
Indeed, it is not hard to see that, upon successful termination, \cref{alg:greedy} constructs a function $f$ that satisfies the assumptions of \cref{thm:proof by contraction}.

In practice, we have found \cref{alg:greedy} to be quite computationally effective.
For all inequalities that we discuss in this paper, the algorithm always terminates successfully by finding a contraction $f$ (even when choices had to be made in the process).
In fact, the analytic proof of the cyclic family of inequalities described in \cref{subsec:cyclic} below has been guided by computer experiments up to $n=11$ regions.
We can provide a computer implementation of \cref{alg:greedy} upon request. 

We remark that there are in general many equivalent ways of expressing an entropy inequality in the form \cref{eq:entropy inequality}.
For example, rescaling a given entropy inequality and/or replacing a right-hand side term $\beta_r S(J_r)$ by a sum $S(J_r) + \dots + S(J_r)$ can often be useful, as it gives more flexibility to the extension process.
From the hypercube perspective discussed above, this amounts to allowing for `finer' cuts $C_r$.

\subsection{A cyclic family of new entropy inequalities}
\label{subsec:cyclic}

We now describe a family of new holographic entropy inequalities.
We have the following infinite family of \emph{cyclic entropy inequalities} for $n \geq 2k+l$ regions,
\begin{equation}
\label{eq:general cyclic}
  \sum_{i=1}^n S(A_i \dots A_{i+l-1}|A_{i+l} \dots A_{i+k+l-1}) \geq S(A_1 \dots A_n),
\end{equation}
where $S(X|Y) := S(XY) - S(Y)$ denotes the \emph{conditional entropy} and where all indices are taken modulo $n$.
We note that the case $(n,k,l)=(2,0,1)$ is subadditivity, $S(A) + S(B) \geq S(AB)$, while the monogamy of the mutual information \cref{eq:mmi} corresponds to choosing $(n,k,l)=(3,1,1)$.
As strong subadditivity can be obtained as a consequence of subadditivity and monogamy, this shows that our family generalizes all previously known holographic entropy inequalities.
The case where $n=2k+1$ and $l=1$ is of particular interest:
\begin{equation}
\label{eq:strongest cyclic}
  \sum_{i=1}^n S(A_i|A_{i+1} \dots A_{i+k}) \geq S(A_1 \dots A_n).
\end{equation}
These are the strongest inequalities in the family, since all other instances of our family can be reduced to \cref{eq:strongest cyclic} and subadditivity. 
Just as the monogamy inequality \cref{eq:mmi} excludes the four-party GHZ state, we can find entangled quantum states $\rho_n$ for each odd $n$ that are excluded by the cyclic inequality \cref{eq:strongest cyclic} for $n$ regions.
Moreover, these states do not violate any of the cyclic inequalities for fewer regions.
This shows that each instance of \cref{eq:strongest cyclic} is independent from the instances for fewer regions. 
Thus it comprises a truly infinite family of cyclic inequalities.
We remark that \cref{eq:strongest cyclic} is in general different from the $n$-partite information, which for $n>3$ is known to not have a definite sign \cite{hayden_headrick_maloney_2013}.

In \cref{sec:cyclic appendix}, we give a detailed proof of our new entropy inequalities and of the properties advertised above.
By exploiting the cyclic structure of the inequalities, we construct an explicit `proof by contraction' $f \colon \{0,1\}^n \rightarrow \{0,1\}^{n+1}$ that extends the initial data given by the occurrence vectors.
While the function has a reasonably simple algebraic description, the cuts \cref{eq:reassembled cuts} produced by it are highly nontrivial (see \cref{sec:cyclic appendix}, where we provide a fully worked-out example).

When all of the regions are adjacent regions on a single boundary, the left-hand side of \cref{eq:strongest cyclic} can be seen as a discrete version of differential entropy \cite{bartek}.
In the $n\rightarrow \infty$ limit, the inequality dictates that the length of the bulk curve tangent to all of the Ryu-Takayanagi surfaces for half the boundary is greater than the entropy of the entire state.
More generally, \cref{eq:general cyclic} together with subadditivity bounds the length an arbitrary convex bulk curve to be greater than the entropy of the state.
In previous work \cite{bartek}, this was shown to follow from strong subadditivity alone, but the proof does not generalize to arbitrary boundary regions.
In contrast, we have shown the general inequality to be true and independent from strong subadditivity.
Our result generalizes an idea such as differential entropy beyond the basic geometric reasoning.

\subsection{The holographic entropy cone for five regions}
\label{subsec:five}

In \cref{subsec:few regions} we had seen that the holographic entropy cone for four regions is completely determined by strong subadditivity and the monogamy of the mutual information alone.
Moreover, we found that all extreme rays could be explained by geometries with sufficiently small boundary cycles (\cref{fig:rays234}) or by star graphs (\cref{fig:raygraphs234}), for which
$S(I) = \min \{ \sum_{i \in I} S(i), \sum_{i \in [n+1] \setminus I} S(i) \}$,
where we set $S(A_{n+1}) := S(A_1 \dots A_n)$ according to purification.
Such entropies are also realized by random pure states with dimension $\log d_i \propto S(i)$ \cite{multiboundary_2014}.

For five or more regions, the situation is markedly different:
There are in general many further entropy inequalities that are independent from strong subadditivity and the monogamy inequality.
One of them is the cyclic inequality \cref{eq:strongest cyclic} for $n=5$ that we have discussed in the preceding section,
\begin{equation}
\label{eq:cyclic five}
  S(A|BC) + S(B|CD) + S(C|DE) + S(D|EA) + S(E|AB) \geq S(ABCDE).
\end{equation}
For example, the permutation-symmetric ray with $S(I) = 1$ for $\abs I \in \{1,5\}$ and $S(I) = 2$ otherwise, which is an extreme ray for $\widehat\calC_5$, is excluded by the cyclic inequality.
Therefore this ray does not correspond to a holographic entropy vector.
The cyclic inequality together with subadditivity and monogamy also implies all stabilizer inequalities for five-partite quantum states derived in
\cite{fivepartystabs}. 
This is in line with our findings in \cref{subsec:few regions} for four regions and it gives non-trivial additional evidence that all holographic entropies are stabilizer entropies.

Apart from \cref{eq:cyclic five}, we have found several other holographic entropy inequalities to be valid, including the following: 
\begin{itemize}
\item $2S(ABC) + S(ABD) + S(ABE) + S(ACD) + S(ADE) + S(BCE) + S(BDE) \geq S(AB) + S(ABCD) + S(ABCE) + S(ABDE) + S(AC) + S(AD) + S(BC) + S(BE) + S(DE)$
\item $S(ABE) + S(ABC) + S(ABD) + S(ACD) + S(ACE) + S(ADE) + S(BCE) + S(BDE) + S(CDE) \geq S(AB) + S(ABCE) + S(ABDE) + S(AC) + S(ACDE) + S(AD) + S(BCD) + S(BE) + S(CE) + S(DE)$
\item $S(ABC) + S(ABD) + S(ABE) + S(ACD) + S(ACE) + S(BC) + S(DE) \geq S(AB) + S(ABCD) + S(ABCE) + S(AC) + S(ADE) + S(B) + S(C) + S(D) + S(E)$
\item $3S(ABC) + 3S(ABD) + 3S(ACE) + S(ABE) + S(ACD) + S(ADE) + S(BCD) + S(BCE) + S(BDE) + S(CDE) \geq 2S(AB) + 2S(ABCD) + 2S(ABCE) + 2S(AC) + 2S(BD) + 2S(CE) + S(ABDE) + S(ACDE) + S(AD) + S(AE) + S(BC) + S(DE)$
\end{itemize}
All the above inequalities, including \cref{eq:cyclic five}, correspond to different facets of the entropy cone and are therefore independent. A feature of these inequalities (and the cyclic family of inequalities) is that they satisfy the following \emph{Bell condition}: they are saturated for Bell pairs shared between any two regions $A_i$ and $A_j$ (including the purifying region).
In particular, this implies that the inequalities are balanced, i.e., each region appears the same number of times on both sides of the inequality.

Another interesting observation is that, in contrast to the situation for four and fewer regions, there are numerous extreme rays of 
$\calC_5$ which cannot be explained by star graphs (see \cref{fig:fromgraph} for one such an example).
Thus the bulk geometry becomes increasingly important; the interior cycles can no longer be ignored for the purposes of minimization in the Ryu-Takayanagi formula \cref{eq:ryu takayanagi}.
We will give a detailed discussion of the holographic entropy cone for five regions in forthcoming work.

\section{The holographic entropy cone and CFT}
\label{sec:physics}

In semi-classical gravity, there is a semi-classical quantum state for each solution of Einstein's equations, possibly coupled to matter fields.
More precisely, there is one such state for each phase space volume measured in units of the Planck constant, and the phase space of the gravity theory is the space of solutions near the initial value surface.

The spacetime metric on the initial value surface cannot be chosen arbitrarily; the metric should satisfy the constraint equations.
Moreover, if we want the Lorentzian-signature geometry to have a smooth analytic continuation to a Euclidean-signature geometry then the extrinsic curvature on the initial value surface should vanish \cite{Halliwell:1989dy}.
This condition ensures that the entanglement entropy for a region on the boundary of the initial value surface is given by the Ryu-Takayanagi formula as opposed to the more general Hubeny-Rangamani-Takayanagi formula \cite{2007JHEP07062H}.
The vanishing of the extrinsic curvature also simplifies the constraint equations.
In particular, if we ignore matter fields and use the vacuum Einstein equations with cosmological constant, the constraint equations reduce to the condition that the scalar curvature of the induced metric on the initial value surface is constant.
Therefore, any $d$-dimensional Riemannian manifold $X$ with constant scalar curvature gives a consistent initial value condition for the vacuum Einstein equations in $d+1$ dimensions.
The resulting solution is time-reversal symmetric in the neighborhood of the initial value surface and thus can be analytically continued to a Euclidean-signature metric, so that the Ryu-Takayanagi formula can be used to evaluate entanglement entropies.
This is why in \cref{subsec:fromgraph} we constructed geometries with constant scalar curvature to realize the extreme rays of our holographic entropy cone.
We remark that constant scalar curvature (together with asymptotically AdS boundary conditions) is sufficient but not necessary for a holographic state to exist; it can be relaxed if we allow matter fields to have non-zero energy-momentum tensors.

Let us examine the properties of such semi-classical quantum states in more detail in the case of $d=2$.
A Riemann surface with a constant negative curvature metric can be lifted to a locally $\AdS_3$ Lorentzian manifold \cite{Brill}.
Specifically, let $\Sigma$ be a Riemann surface with constant curvature metric $\mathrm{d}\Sigma^2$ and curvature scale $\ell$.
Then there exists a locally $\AdS_3$ space $\mathcal{M}_\Sigma$ having a coordinate patch in a neighborhood of $t=0$ that looks as follows:
\begin{equation}
\mathrm{d}s^2_{\mathcal{M}_\Sigma} = - \ell^2 \mathrm{d}t^2 + \cos^2{t} \, \mathrm{d} \Sigma^2 .
\label{eq:metric}
\end{equation}
Note that the manifold  $\mathcal{M}_\Sigma$ includes the entirety of the Riemann surface on a time-reflection symmetric slice at $t=0$.
Such locally $\AdS_3$ solutions to Einstein's equations are realized -- after removing certain regions -- by quotients of $\AdS_3$ by Fuchsian groups of the second kind \cite{skenderis_van_rees_2011,Brill}.
The desired group action can be uniquely identified by its action on the $t=0$ slice alone.
It is precisely the group action on the hyperbolic disk $\mathbb H^2$ whose quotient is our Riemann surface $\Sigma$.
A thorough discussion of these solutions can be found in \cite{skenderis_van_rees_2011} and a concise, but very readable treatment in \cite{multiboundary_2014}, which mirrors our present interests.

If our Riemann surface $\Sigma$ has $b$ boundaries, the Lorentzian geometry will also have $b$ asymptotic locally $\AdS_3$ regions.
Each exterior region has the exact geometry of a BTZ black hole.
These asymptotic BTZ regions are connected by a wormhole hidden behind the collective BTZ horizons.
Each asymptotic boundary corresponds to a separate CFT, not directly coupled to the other CFTs.
A wormhole with $b$ boundaries is dual to a state $\ket\Sigma$ in the tensor product $\mathcal{H}_{CFT}^{\otimes b}$ of the Hilbert spaces of the individual CFTs.
The interior regions of the Riemann surface encode the $t=0$ slice of the wormhole as well as the entanglement structure of the state $\ket\Sigma$. 

The Lorentzian geometry that we have obtained from our Riemann surface has a time-reversal symmetry.
Thus there is a well-defined Euclidean continuation of the metric \cref{eq:metric}, which in our case is given by
\begin{equation}
\mathrm{d}s^2 =  \ell^2 \mathrm{d}\tau^2 + \cosh^2{\tau} \, \mathrm{d} \Sigma^2  .
\label{eq:eumetric}
\end{equation}
Because the time-symmetric surface is precisely the Riemann surface $\Sigma$, the corresponding entanglement entropies can be obtained by the Ryu-Takayanagi formula from the minimal geodesics in $\Sigma$ \cite{2007JHEP07062H,multiboundary_2014}.
Thus, our holographic entropy vectors are a true measure of the boundary entanglement entropies.

\begin{figure}
\centering
\includegraphics[width=0.4\linewidth]{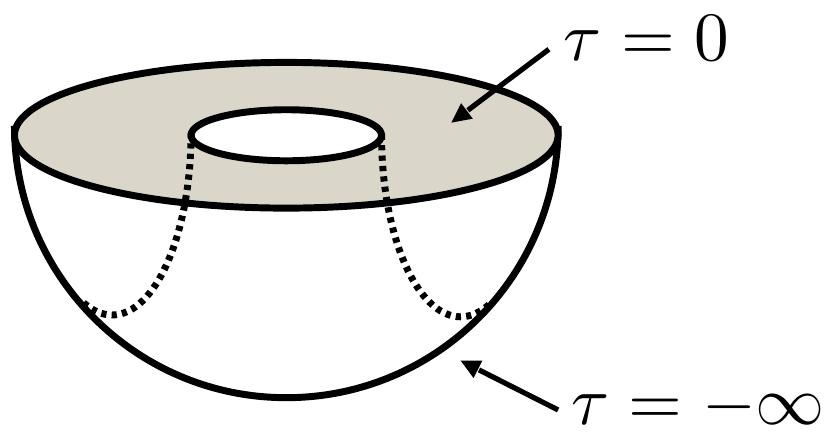}
\caption{Euclidean geometry connecting two CFT states. The shaded region is the time-symmetric surface $\Sigma$
at $\tau=0$. Note that the asymptotic boundary at $\tau=-\infty$ is also a copy of $\Sigma$.}
\label{fig:wormhole1}
\end{figure}

Since the Euclidean geometry \cref{eq:eumetric} matches onto the Lorentzian
geometry at $\tau = 0$, where the extrinsic curvature vanishes,
the Euclidean geometry can be regarded as a saddle point of the Euclidean
path integral of the bulk gravity and defines the state $|\Sigma \rangle$
at $\tau=0$ in the semi-classical gravity approximation.
On the other hand, the metric \cref{eq:eumetric} becomes
$\mathrm{d}s^2 \sim \ell^2 \mathrm{d}\tau^2 + e^{2|\tau|} \, \mathrm{d} \Sigma^2$
as $\tau \rightarrow -\infty$, and therefore the asymptotic boundary
is  also a copy of $\Sigma$ with $e^{|\tau|}$ as its RG scale, as shown in \cref{fig:wormhole1}. Thus,
the Riemann surface $\Sigma$ is a natural home for the dual CFT at $\tau = -\infty$ as well
as the time symmetric slice at $\tau = 0$ of the bulk geometry.
Based on this observation,
it was argued in \cite{multiboundary_2014} that the state $|\Sigma \rangle$ is given by
a functional integral of the $2d$ CFT on $\Sigma$.

States of two-dimensional CFTs on a Riemann surface with boundaries were considered
in the 1980s in \cite{OperatorOne, OperatorTwo, OperatorThree} to formulate the
operator formalism for conformal field theory on Riemann surfaces of higher genus.
For example, for the $c=1$ CFT with complex free fermions $\psi(z), \bar{\psi}(z)$,
we consider a $(-1/2, 0)$ form $\nu(z)$ that is holomorphic everywhere on
a Riemann surface $\Sigma$ with $b$ boundaries $\gamma_i$ ($i=1,...,b$)
obeying the same spin structure as $\psi(z)$. For simplicity, let us assume
NS spin structure around each boundary $\gamma_i$.
We then require that, for any such $\nu(z)$, the state
$|\Sigma \rangle$ in the $b$ tensor product of the NS fermion Fock spaces satisfies
\begin{equation}
 \sum_{i=1}^b \oint_{\gamma_i} dz\ \nu(z) \psi^{(i)} (z) |\Sigma \rangle = 0,  \quad
 \sum_{i=1}^b \oint_{\gamma_i} dz\ \nu(z) \bar{\psi}^{(i)} (z) |\Sigma \rangle = 0,
\label{eq:grassmann}
\end{equation}
where $\psi^{(i)} (z)$ is the free fermion operator
acting on the $i$-th Hilbert space
of the CFT. It turns out that, for each spin structure,
there is a unique solution to these equations.
Moreover, the resulting state $|\Sigma \rangle$ can be used to compute
correlation functions of these fermions as
\begin{align*}
&\braket{\psi(z_1) \cdots \psi(z_N) \bar{\psi}(w_1) \cdots \bar{\psi}(w_N)}_{\Sigma_0} \\
=\;&\prescript{}{1}{{\bra 0}} \otimes \cdots \otimes \prescript{}{b}{{\bra 0}} \prod_{\alpha=1}^N \psi^{(i_\alpha)}(z_\alpha) \prod_{\beta=1}^N \bar{\psi}^{(i_\beta)}(w_\beta)
 \ \ket\Sigma,
\end{align*}
where the left-hand side is the correlation function computed on a closed
Riemann surface $\Sigma_0$ obtained by capping off the $b$ boundaries of $\Sigma$
by attaching a disk to each of them. The conditions \cref{eq:grassmann} ensure
that fermions are well-defined and holomorphic on $\Sigma$.
The state $\ket\Sigma$ associated to the Riemann surface $\Sigma$ plays important
roles in soliton theory \cite{soliton} and in topological string theory \cite{topological_vertex}.

It is an interesting question to ask if a given entangled state in  $\mathcal{H}_{CFT}^{\otimes b}$ is interpolated through a smooth wormhole geometry.
In the case at hand, it is possible to diagnose whether a given state $\ket{\Psi^{(b)}}$ is associated to a smooth Riemann surface with $b$ boundaries.
For $b=1$ the procedure is as follows.
Consider the fermion current $J(z) = : \psi(z) \bar{\psi}(z) :=\sum_n J_n z^{-n-1}$ and define the following wave-function
of infinitely many variables $t_1, t_2, \dots$,
\[ \Psi(t) = \bra 0 \exp\left(\sum_{n=1}^\infty t_n J_n  \right) \ket\Psi. \]
If $\ket\Psi = \ket\Sigma$ for some Riemann surface $\Sigma$ of genus $g$ with one boundary,
we can compute $\Psi(t_1, t_2, \cdots)$ explicitly as
\[ \Psi(t) = \exp\left(\sum_{m,n} Q_{mn} t_m t_n\right)\vartheta\left( \sum_n \vec{A}_n t_n + \vec{c} \right), \]
where $\vartheta$ is the Riemann theta function for $\Sigma$. The matrix $Q_{mn}$ is related to the current-current correlation function,
and the $g$-dimensional vectors $\vec{A}_n$ and $\vec{c}$ are related to the holomorphic one-forms and spin structure on $\Sigma$, respectively.
Note that modulo the quadratic form $\sum_{m,n} Q_{mn} t_m t_n$ in the exponent, $\Psi$ depends on the $t_n$ only through the $g$ linear combinations
$\sum_n \vec{A}_n t_n$
(equivalently, we can say that $\partial^3\log \Psi/\partial t_k \partial t_m \partial t_n$ depends only on these combinations of $t_n$).
Conversely, it was proven in  \cite{mulase1984} that if $\Psi$ depends on the $t_n$ only through $g$ linear combinations,
modulo a quadratic form in the exponent, there is a Riemann surface $\Sigma$ with one boundary such that
$\ket\Psi = \ket\Sigma$. To generalize this for $b \geq 1$ we can consider
\[ \Psi(t) = \prod_{i=1}^b \prescript{}{i}{{\bra 0}} \exp\left(\sum_{n=1}^\infty t_n^{(i)} J_n^{(i)}  \right) \ket\Psi, \]
where the $J_n^{(i)}$ are the current operators acting on the $i$-th Hilbert space.
The state $\ket\Psi$ is associated to a Riemann surface of genus $g$ with $b$ boundaries if and only if the wave-function $\Psi(t)$ depends on the $t_n^{(i)}$ only through $g$ linear combinations modulo a quadratic form in the exponent.
A natural generalization of \cref{eq:grassmann} for a general
CFT would be
\begin{equation}
 \sum_{i=1}^b \oint_{\gamma_i} dz\ \xi(z) T^{(i)} (z) |\Sigma \rangle = 0
\label{eq:virasoro}
\end{equation}
and its complex conjugation,
where $\xi(z)$ is any $(-1,0)$ form holomorphic on $\Sigma$ and
$T^{(i)} (z)$ is the energy-momentum tensor
acting on the $i$-th Hilbert space
of the CFT. These conditions ensure that the energy-momentum tensor
is well-defined and holomorphic on $\Sigma$.
Unlike in the case of free fermions,
a solution to these equations is not unique in general.
In fact, there is a solution for each conformal block of the CFT on $\Sigma$.
Additional conditions are required to specify the state $\ket\Sigma$ uniquely.

To illustrate these points, let us consider as $\Sigma$ a sphere with two circular holes.
By Weyl rescaling and diffeomorphisms, we can turn $\Sigma$ into a cylinder of length $\beta/2$.
The two boundary circles can be twisted relative to each other by $\theta/2$.
Since $\ket\Sigma$ is a state in $\mathcal{H}_{CFT}\otimes \mathcal{H}_{CFT}$, one can also think of it as a map from $\mathcal{H}_{CFT}$ to $\mathcal{H}_{CFT}$.
The functional integral of the two-dimensional CFT on $\Sigma$ gives the Euclidean time translation and rotation, $\exp(-\frac{\beta}{2} \hat{H} -i\frac{\theta}{2}\hat{J})$, where $\hat{H}$ is the Hamiltonian and $\hat{J}$ is the angular momentum around the cylinder.
The corresponding state $\ket\Sigma$ is then given by
\begin{equation}
  \ket\Sigma = \sum_{\Psi} e^{-\frac{\beta}{2} E_\Psi -i\frac{\theta}{2}J_\Psi}
  \ket\Psi \otimes \ket\Psi,
\label{eq:thermo_field_double}
\end{equation}
where the $\ket\Psi$ denote the joint energy and angular momentum eigenstates in $\mathcal{H}_{CFT}$, with $E_\Psi$ and $J_\Psi$ the corresponding eigenvalues.
This is the thermofield double state.
However, this is not the only solution to \cref{eq:virasoro}. In fact, any projection of $\exp(-\frac{\beta}{2} \hat{H} -i\frac{\theta}{2}\hat{J})$ to
a Virasoro primary state $\phi$ and its descendants satisfies the condition. The corresponding state is given by
\begin{equation}
  \ket\Sigma_\phi = \sum_{i_L, i_R} e^{-\frac{\beta}{2} (h_\phi + n_{i_L} + \bar{h}_\phi+ n_{i_R} -\frac{c}{12}) -i\frac{\theta}{2}(h_\phi + n_{i_L}-\bar{h}_\phi - n_{i_R})}
  \ket{\phi; i_L, i_R} \otimes \ket{\phi; i_L, i_R},
\label{eq:projected_to_descendants}
\end{equation}
where the $\ket{\phi; i_L, i_R}$ are normalized states for the left and right Virasoro descendants of $\phi$; their conformal weights are $(h_\phi + n_{i_L} ,  \bar{h}_\phi+ n_{i_R})$.
We included the Casimir energy $-\frac{c}{12}$ of the CFT on the circle.
These states form a complete basis of solutions for
\cref{eq:virasoro}. For example, the
thermofield double state \cref{eq:thermo_field_double} can be expressed as the
sum over all primary states with uniform weight,
\begin{equation}
\ket\Sigma = \sum_{\phi : \text{primary}}
\ket\Sigma_\phi .
\end{equation}

\begin{figure}
\centering
\includegraphics[width=0.3\linewidth]{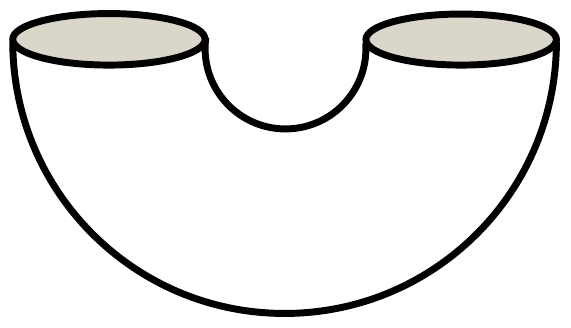}
\caption{Thermal $\AdS_3$, which is another saddle-point geometry we can use to define a state in  $\mathcal{H}_{CFT}\otimes \mathcal{H}_{CFT}$ solving
\cref{eq:virasoro}.}
\label{fig:wormhole2}
\end{figure}

Returning to (2+1)-dimensional gravity, there are infinitely many Euclidean saddle-point
geometries that can connect $\mathcal{H}_{CFT}$ to $\mathcal{H}_{CFT}$. The geometry \cref{eq:eumetric}, with $\Sigma$ being the cylinder
as shown in \cref{fig:wormhole1},
is one possibility. But, we may also consider
a solid cylinder in which the circle direction is contractible and fill in with a portion of the Euclidean $\AdS_3$, as shown
in \cref{fig:wormhole2}. There are infinitely many other possibilities as well. Each saddle-point geometry $\mathcal M$ such that
$\partial {\mathcal M} = \Sigma$ gives a semi-classical state
$\ket{\mathcal M}\rangle$ in $\mathcal{H}_{CFT}\otimes \mathcal{H}_{CFT}$ which
satisfies \cref{eq:virasoro} and therefore is a linear combination of
$\ket{\Sigma}_\phi$ given by \cref{eq:projected_to_descendants},
\begin{equation}
 \ket{\mathcal M}\rangle = \sum_{\phi:\text{primary}} C_\phi({\mathcal M}) \ \ket\Sigma_\phi,
 \label{eq:primary_weights}
\end{equation}
where the weights $C_\phi({\mathcal M})$ depend on the choice of the Euclidean geometry $\mathcal M$.
We can determine the weights by computing the norm of $\ket{\mathcal M}$,
\begin{align*}
  \braket{\mathcal M | \mathcal M}
=\;&\sum_{\phi:\text{primary}} \abs{C_\phi({\mathcal M})}^2 \ {}_\phi\braket{\Sigma|\Sigma}_\phi \\
=\;& q^{-\frac{c}{24}}\bar{q}^{-\frac{c}{24}} \prod_{n=1}^\infty \abs{1-q^n}^{-2} \times
\left( \abs{1-q}^2 \abs{C_{{\rm vac}}({\mathcal M})}^2 +
\sum_{\phi\neq {\rm vac}}q^{h_\phi}\bar{q}^{\bar{h}_\phi}
\abs{C_\phi({\mathcal M})}^2 \right),\nonumber
\end{align*}
where $q = e^{-\beta + i \theta}$ and ${\rm vac}$ denotes the ground state.
We used the fact that descendants of the ground state contribute to the
partition function as
$q^{-\frac{c}{24}}\bar{q}^{-\frac{c}{24}}\prod_{n=2}^\infty \abs{1-q^n}^{-2}$
while those for a generic primary state $\phi$ contribute as
$q^{h_\phi-\frac{c}{24}}\bar{q}^{\bar{h}_\phi-\frac{c}{24}}
\prod_{n=1}^\infty \abs{1-q^n}^{-2}$.
From the point of view of (2+1)-dimensional semi-classical gravity,
the norm $\braket{\mathcal M | \mathcal M}$ is the partition function
computed at the Euclidean saddle geometry obtained by gluing $\mathcal{M}$
with itself. The asymptotic boundary of the resulting geometry is a torus, which
is the Schottky double of the cylinder. We can then find the weights $C_\phi$
by comparing the semi-classical gravity
partition function with the CFT norm in the above.

For example, if $\mathcal M$ is the solid cylinder, the Euclidean saddle for
$\braket{\mathcal M | \mathcal M}$ is thermal $\AdS_3$. In this case,
the primary states contributing to the partition function are the ground state, a few primary
states corresponding to low energy particles in $\AdS_3$, and their Virasoro descendants.
Thus, we have $C_{{\rm vac}}=1$ for the vacuum,
and $C_{\phi}=1$ when $\phi$
corresponds to a light particle in $\AdS_3$. Otherwise, $C_{\phi}=0$.

For the wormhole geometry \cref{eq:eumetric} we are interested in,
the Euclidean saddle $\braket{\mathcal M | \mathcal M}$ is 
given by the metric \cref{eq:eumetric} with $-\infty < \tau < \infty$.
It is the Euclidean black hole with temperature $1/\beta$ and chemical potential $\theta$ for the
angular momentum, which is the modular transform of thermal $\AdS_3$. Thus,
the weights $C_\phi({\mathcal M})$ for this geometry should satisfy
\begin{align*}
&q^{-\frac{c}{24}}\bar{q}^{-\frac{c}{24}} \prod_{n=1}^\infty \abs{1-q^n}^{-2} \times
\left( \abs{1-q}^2 \abs{C_{{\rm vac}}({\mathcal M})}^2 +
\sum_{\phi\neq {\rm vac}}q^{h_\phi}\bar{q}^{\bar{h}_\phi}
\abs{C_\phi({\mathcal M})}^2 \right) \nonumber \\
=\;&\tilde{q}^{-\frac{c}{24}}\bar{\tilde{q}}^{-\frac{c}{24}}
\prod_{n=1}^\infty \abs{1-\tilde{q}^n}^{-2} \times
\left( \abs{1-\tilde{q}}^2 +
\sum_{\phi:\text{particles}}\tilde{q}^{h_\phi}\bar{\tilde{q}}^{\bar{h}_\phi} \right),
\end{align*}
where $\tilde{q}=e^{\frac{(2\pi)^2}{-\beta + i\theta}}$. Since we are in the regime of
semi-classical gravity, the central charge $c$ is large. In this case,
an approximate solution to this equation is given by the Cardy formula:
\[
\abs{C_\phi({\mathcal M})}^2 \simeq
\exp\left[2\pi \sqrt{\frac{c}{6}\left(h_\phi - \frac{c}{24}\right)}+2\pi \sqrt{\frac{c}{6}\left(\bar{h}_\phi - \frac{c}{24}\right)}\right]
\]
The CFT state dual to the wormhole with two boundaries is then given by \cref{eq:primary_weights} with these weights.

It would be interesting to similarly identify the CFT states for multiboundary wormholes,
and we leave this for future investigation.

\paragraph{CFT states and the entropy cone.}

To prove that the holographic entropy vectors live in a convex cone we argued that the space of geometries was closed under scaling by a positive constant and under disjoint union of the geometries.
We briefly elaborate on those operations in terms of the dual CFT.
The entanglement entropy is related to the area of the Ryu-Takayanagi surface by $S=A/4G_N$.
It is thus, more precisely, the ratio of the area over the gravitational coupling, $A/G_N$, that we wish to scale.
For the CFT state, rescaling the bulk metric by a positive constant will rescale the asymptotic value of the AdS radius $\ell/G_N$, which corresponds to rescaling the CFT central charge $c= 3\ell/2G_N$.
Such a rescaling is a natural operation and is generically allowed for large $c$ CFTs with bulk dual.
The disjoint union operation is even simpler.
Taking a disjoint union of geometries $\Sigma' = \Sigma_1 \cup \Sigma_2$ corresponds to identifying $\mathcal{H}^{i}_{CFT'}=\mathcal{H}^{i}_{CFT_1} \otimes \mathcal{H}^{i}_{CFT_2}$ and considering the tensor-product state $\ket{\Sigma'} = \ket{\Sigma_1} \otimes \ket{\Sigma_2}$.


\section{Outlook}
\label{sec:outlook}

We conclude by sketching a number of interesting avenues for future investigation motivated by our results.

\paragraph{ER$=$EPR.}

While there is ample evidence that there are important connections between entanglement and geometry, many of these connections are currently only vaguely understood.
Recent work by Susskind and Maldacena \cite{lenny} has suggested that there is an equivalence between Einstein-Rosen wormholes (ER) and
entangled Einstein-Rosen-Podolsky pairs (EPR).
While this seems to be a sufficient language to discuss two-sided black holes, there is some evidence that a more general framework is needed to discuss multi-partite entanglement.
Previous work by Balasubramanian et al.\ \cite{multiboundary_2014} has argued that there exist $n$-boundary black holes dual to CFT states with `intrinsically $n$-partite' entanglement.
That is, the states 
cannot be prepared by taking a tensor product of states that are strictly less than $n$-partite, suggesting that geometric wormholes are not best understood in terms of bipartite EPR pairs alone.
Entropy cones provide a systematic framework for characterizing the entropies of multipartite entangled states.
We may stratify the holographic entropy cones according to the minimal entanglement required to explain a given entropy vector.
Thus the $k$-th stratum contains those holographic entropy vectors that can only be explained by $k$-partite (or higher) entanglement.
By embedding quantum entropy cones for fewer regions as described in \cref{subsec:symmetries}, we may then obtain entropy inequalities for each $k$, whose violation demonstrates that the entanglement is at least $(k+1)$-partite.
The inequalities proved in \cite{multiboundary_2014} are one such class of inequalities, though they are not the strongest possible.
Similar ideas have been proposed previously on the level of entanglement spectra \cite{walter2013entanglement}.

There are also intrinsically $n$-partite entangled states that cannot be detected entropically.
This invites a more refined study of multipartite holographic correlations and their decomposition into fundamental building blocks.


\paragraph{Quantum gravity.}

It would be interesting to verify which of the holographic entropy inequalities continue to hold beyond the classical level when subleading corrections are taken into account. Indeed, while the Ryu-Takayanagi formula describes the entropy to leading order in the central charge, there are situations where it is known that the entanglement entropy is non-zero only to subleading order. More generally, there exist geometries for which all leading order terms cancel in an entropy inequality; these correspond to points on the facets of our holographic entropy cones. A natural question would thus be to determine how subleading entropy corrections modify the facets of the cones (see, e.g., \cite{barrella2013holographic,2013JHEP11074FLM}).

\paragraph{Covariant holographic entropy.}
It would also be interesting to extend the methodologies developed in this work to covariant descriptions of boundary entanglement entropies, such as the Hubeny-Rangamani-Takayanagi proposal \cite{2007JHEP07062H}.
The work \cite{wall2014maximin} implies that at least under certain conditions on the structure of an inequality, it is possible to reduce to individual time slices, so that our proofs by contraction become applicable.
More generally, it would be of physical interest to determine whether the constraints imposed on covariant entropies by the Hubeny-Rangami-Takayanagi formula are strictly weaker than those imposed by the Ryu-Takayanagi formula.

\paragraph{Tensor networks.}

The entropy inequalities that we consider in this paper hold for any physical system that satisfies a version of the Ryu-Takayanagi formula or the discrete entropy formula -- such as MERA states built from generic tensors \cite{evenbly2011tensor} and the tensor networks constructed recently in \cite{happy}.
They also serve as additional checks for tensor network proposals for the gauge/gravity correspondence. 
In \cref{subsec:few regions,subsec:five} we have found that any holographic entropy vector is compatible with the known constraints on the entropies of stabilizer states for $n \leq 5$ subsystems.
This is rather suggestive in view of the recently proposed connections between holography and quantum error correcting codes \cite{verlinde_verlinde_2013,almheiri_dong_harlow_2014}, and it is likely that this observation can be proved for arbitrary $n$ by generalizing the construction of \cite{happy}.

Likewise, bulk inequalities \cref{eq:bulk inequality quote} naturally hold for all physical systems that saturate an area law, $S(A) \sim \abs{\partial A}$. It would be interesting to explore the implications of our new inequalities for such systems, in particular since in this context the tripartite information (i.e., difference of the left- and right-hand side of the monogamy inequality) specializes to the topological entanglement entropy \cite{kitaev2006topological}.

\paragraph{Geometry, combinatorics, and complexity.}

In differential geometry, it has been of long-standing interest to characterize the properties of geodesic length functions. This has been fully worked out in \cite{luo1998geodesic}, where it was shown that the constraints amount to a set of polynomial inequalities.
Our results show that the properties of geodesic length functions are much simpler when they are optimized over homology classes; in this case there is a finite number of \emph{linear} constraints. It would be desirable to extend this to a complete characterization along the lines of \cite{luo1998geodesic}.

In graph theory, the cut function, which assigns to a cut the total weight of all edges that go across the cut, has long been known to be \emph{submodular}, i.e., to satisfy strong subadditivity.
As discussed in \cref{subsec:contractions}, this is in a precise technical sense the discrete version of strong subadditivity of the Ryu-Takayanagi formula.
Submodular functions can be regarded as discrete analogues of convex functions and they play a natural role in combinatorial optimization.
Our bulk inequalities \cref{eq:bulk inequality quote} show that the cut function satisfies many further linear inequalities than were previously known.
In contrast, not even monogamy appears to be known in this context, and it should be interesting to study the consequences.

Finally, we recall from \cref{subsec:contractions} that the problem of finding proofs by contraction is equivalent to finding extensions of certain partially defined graph homomorphisms of hypercube graphs.
To decide whether a given partially defined function can be extended to a graph homomorphism is a natural problem in graph theory and it would be desirable to determine the computational complexity of this decision problem.

\begin{acknowledgments}
We thank Mario Berta, Venkat Chandrasekaran, Bartek Czech, Patrick Hayden, Shaun Maguire, Alexander Maloney, Donald Marolf, Ingmar Saberi, and Adam Sheffer for pleasant discussions.
M.W.\ acknowledges funding provided by the Simons Foundation and FQXi. N.B., H.O., and B.S.\ are supported in part by U.S.\ DOE grant DE-SC0011632 and by the Walter Burke Institute for Theoretical Physics (Burke Institute) at Caltech. The work of H.O.\ is also supported in part by the Simons
Investigator Award, by the WPI Initiative of MEXT of Japan, and by JSPS Grant-in-Aid for Scientific Research C-26400240. N.B.\ is a DuBridge Fellow of the Burke Institute. S.N.\ is supported by a Stanford Graduate Fellowship.
N.B.\ and B.S.\ thank the Stanford Institute for Theoretical Physics for hospitality. S.N., J.S., and M.W.\ thank the Burke
Institute at Caltech for hospitality.
H.O.\ thanks the hospitality of the Simons Foundation at the Simons Symposium on Quantum Entanglement.
\end{acknowledgments}

\appendix

\section{Universal graph models for holographic entropies}
\label{subsec:universal graph models}

In this section, we discuss an alternative approach to constructing graph models of holographic entropies.
Previously, in \cref{subsec:tograph}, we had described a method for constructing a graph for any geometry and assignment of boundary regions.
Our approach in this section is somewhat different.
Here, we discretize the given geometry to obtain a fine-grained graph model such that, for \emph{arbitrary} boundary regions in the geometry and corresponding boundary regions in the graph model, the discrete entropy in the graph model (defined as in \cref{dfn:graph model} to be the weight of a minimum cut separating boundary region and its complement) approximates the Ryu-Takayanagi entropy in the geometry.

The fact that the whole geometry can be converted to large graph with the same entropic structure is conceptually interesting, since it can be thought as a step towards constructing more advanced models of the bulk/boundary correspondence that capture the whole system at once (without requiring an a priori choice of boundary regions).
For example, one can think of our graph model as a skeleton that a tensor network model can be put on top of.
Furthermore, we will see that this approach naturally leads to a description of the geometry as a finely-weaved network of multiboundary wormholes and thereby supports the picture laid out by the ER$=$EPR proposal \cite{lenny} (cf.\ \cref{sec:outlook}).

In the following we sketch the method for two-dimensional surfaces (however, we believe that the method can be extended to higher dimensions and made mathematically rigorous).
To start, consider a Riemannian manifold with boundary (see \cref{fig:fromgraph-universal}, (a)).
Any triangulation of the bulk manifold gives rise to a weighted graph, which we call a \emph{discretization} of the geometry, where the weight of an edge in the triangulation is defined to be the geodesic distance of its endpoints in the geometry.
Let us choose a very fine-grained triangulation of the bulk manifold.
This triangulation should be such that any tangent vector of the manifold can be well-approximated by an edge of the graph.
Then we can approximate the geodesic distance between any two points on the geometry by the weight of a shortest path connecting the corresponding two points on the graph.
For example, the length of the minimal geodesic for region $A$ in \cref{fig:fromgraph-universal}, (a) is well-approximated by the length of the path highlighted in the graph in \cref{fig:fromgraph-universal}, (b).%
\footnote{Although $A$ is a full boundary in \cref{fig:fromgraph-universal}, (a), in general there is no need to assume so.}
The final step is to construct the dual graph of the discretization.
This is the graph obtained by adding vertices for all triangles of the original graph and edges between any two neighboring triangles.
The weight of an edge in the dual graph is defined as the weight of the edge separating the triangles in the original graph, divided by $4 G_N$ (see \cref{fig:dualgraph} and \cref{fig:fromgraph-universal}, (c)).
Indeed, while in the discretization the edge weights carry units of length, the edge weights in the dual graph represent entropy.
We remark that the obtained graph is automatically trivalent at its bulk vertices.
Now observe that any path on the original discretization of the geometry cuts a number of edges of the dual graph along its way.
The total weight of these edges is by definition proportional to the length of the path.
As an example, in \cref{fig:fromgraph-universal}, (c) we have highlighted the edges of the dual graph that are cut by the path in \cref{fig:fromgraph-universal}, (b). The total length of the former, divided by $4 G_N$, equals the weight of the latter, which in turn approximates the entropy of region $A$.
Thus it is easy to see that in general the weight of a minimal cut separating an arbitrary boundary region $A$ from its complement is equal to the length of a shortest path on the original graph that is homologous to $A$, and therefore approximately equal to the entropy of $A$.
We conclude that the dual graph constructed in this way approximates the holographic entropies of the given bulk geometry for any choice of boundary region.

\begin{figure}
  \centering
  \includegraphics[width=0.7\linewidth]{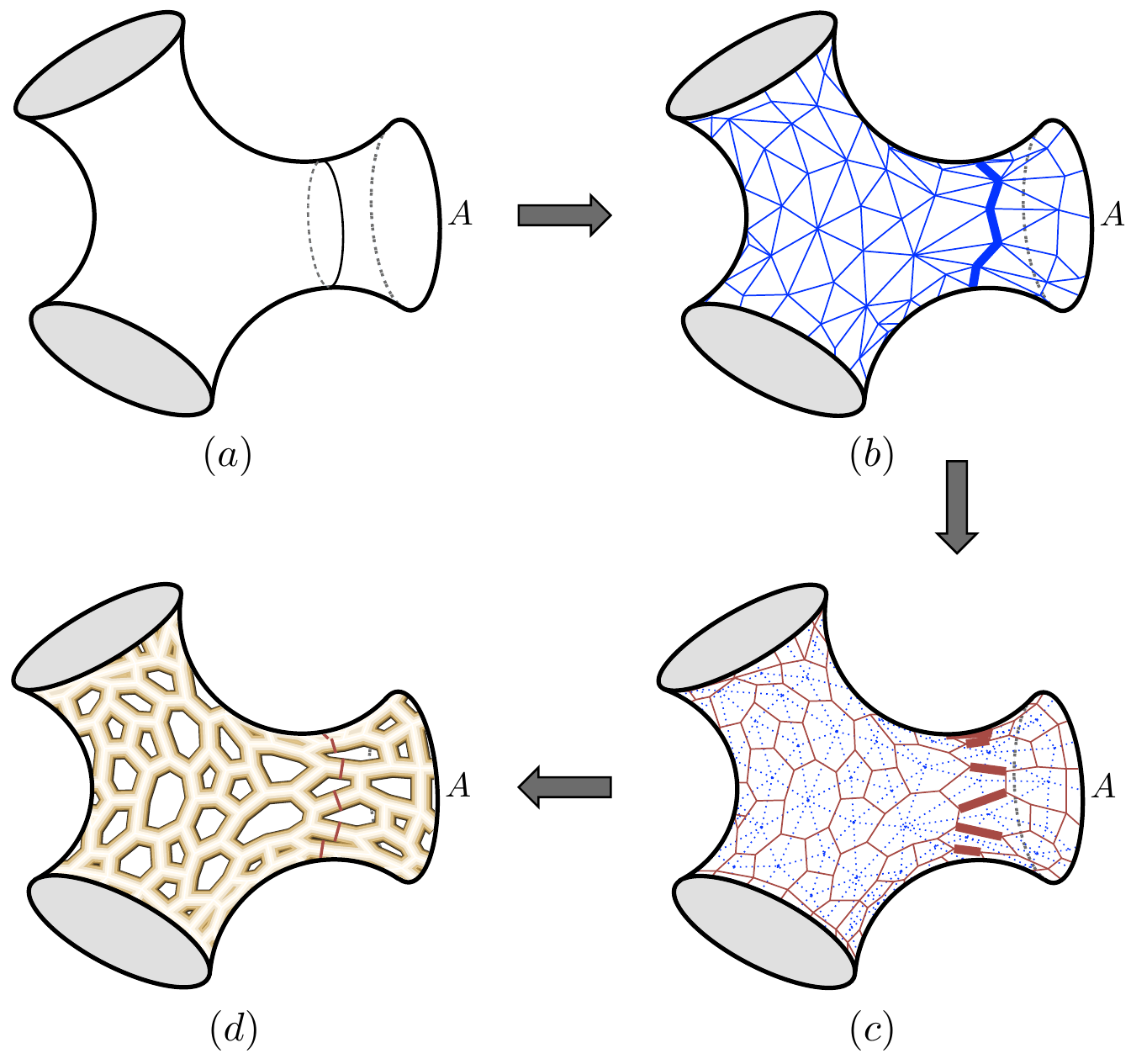}
  \caption{(a)--(c) Construction of a universal graph model obtained by discretizing the given geometry; (d) corresponding multiboundary wormhole geometry.}
\label{fig:fromgraph-universal}
\end{figure}

\begin{figure}
  \centering
  \includegraphics[width=0.7\linewidth]{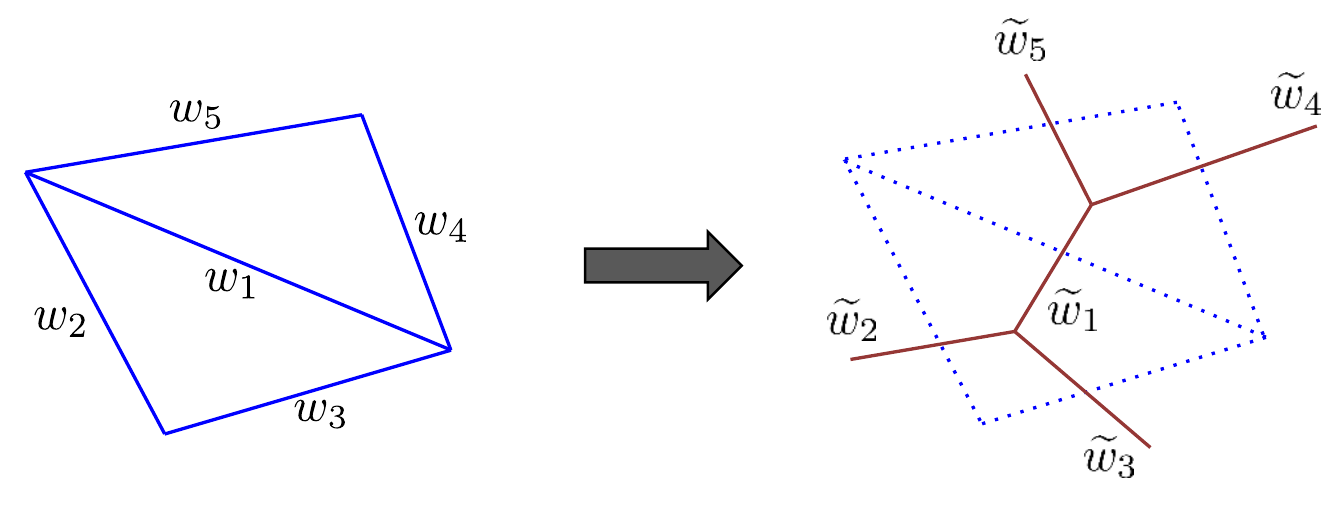}
  \caption{Construction of the dual graph. The weight $\widetilde w$ of any dual edge is defined as the weight $w$ of the edge it cuts in the original graph divided by $4 G_N$.}
\label{fig:dualgraph}
\end{figure}

We remark that a similarly fine-grained discretization can also be obtained by the procedure of \cref{lem:tograph} when applied to a very fine decomposition of the boundary manifold, $\partial X = A_1 \cup \dots \cup A_n$ (which also works for dimensions higher than two).

\medskip

It is interesting to go further and convert this graph model back into a hyperbolic surface as described in \cref{subsec:fromgraph}.
In this way we obtain a finely-weaved network of multiboundary wormhole geometries that has the same entropic structure of the original bulk surface.
See \cref{fig:fromgraph-universal}, (d) for an illustration, where we have also highlighted the minimal geodesic for region $A$.
Its length approximates the length of the minimal geodesic in the original geometry, \cref{fig:fromgraph-universal}, (a), and so the holographic entropies are approximately equal.

\section{Proof of the cyclic entropy inequalities}
\label{sec:cyclic appendix}

In this section we give a detailed proof of the cyclic family of holographic entropy inequalities \cref{eq:general cyclic} that we discussed in \cref{subsec:cyclic}.
For this, it will be convenient to introduce the notation $C_n(k) := \sum_{i=1}^n S(A_{i+1} \dots A_{i+k})$ for a cyclic sum of $k$-body entropies of an $n$-partite system.
Then \cref{eq:general cyclic} can be written in the following way:
\begin{equation}
\label{eq:general cyclic appendix}
  C_n(a, b) := C_n(a) - C_n(b) - S (A_1 \dots A_n) \geq 0
\end{equation}
where $a = k + l$ and $b = k$ in terms of the parametrization of \cref{eq:general cyclic}.
By evaluating \cref{eq:general cyclic appendix} on Bell pairs, we obtain the following necessary conditions on the parameters $a$, $b$ and $n$ for \cref{eq:general cyclic appendix} to be valid holographic entropy inequality:

\begin{lem}
\label{lem:cyclic necessary}
  The cyclic inequality \cref{eq:general cyclic appendix} can only be valid if $a > b$ and $a+b \leq n$.
\end{lem}
\begin{proof}
  For the first assertion, consider a geometry where all regions except $A_1$ are trivial.
  Entropically, this amounts to a Bell pair shared between $A_1$ and the purifying region $A_{n+1}$.
  Thus $C_n(a,b) \geq 0$ reduces to
  \[ a S(A_1) - b S(A_1) - S(A_1) \geq 0, \]
  which implies that $a - b - 1 \geq 0$, or $a > b$.

  For the second assertion, consider the entropies of a Bell pair shared between $A_1$ and $A_k$ (cf.~\cref{subsec:few regions}, where we had seen that these entropies are geometric). If we choose $k = \lfloor\frac n2\rfloor + 1$ then it is not hard to verify that
  $C_n(k) = 2 \min \{ k, n-k \}$.
  Thus $C_n(a,b) \geq 0$ amounts to
  \[ 2 \min \{ a, n-a \} - 2 \min \{ b, n-b \} \geq 0. \]
  If $a \leq \frac n2$ then $a + b < 2a \leq n$ by the first part, which is what we wanted to show.
  Otherwise, if $a > \frac n2$, then together with the first part we obtain the condition
  \[ n - b > n - a \geq \min \{ b, n-b \}. \]
  This implies that $n - a \geq b$, or $a + b \leq n$, as we set out to show.
\end{proof}
In the remainder of this section, we will show that the necessary conditions established in \cref{lem:cyclic necessary} are also sufficient.

\paragraph{Reduction.}
As a first step, we will apply a series of reductions. For this, we need the following lemma: 

\begin{lem}
\label{lem:reduction}
  The following inequalities are a direct consequence of strong subadditivity and weak monotonicity of the holographic entropy:
  \begin{itemize}
  \item For all $a \geq b \geq 1$, we have that $C_n(a-1,b-1) \geq C_n(a,b)$.
  \item For all $a + b \leq n$, we have that $C_n(a+1, b-1) \geq C_n(a, b)$.
  \item For all $a < \frac n2$, we have that $C_n(a+1) \geq C_n(a)$.
  \end{itemize}
\end{lem}
\begin{proof}
  For the first claim, we use strong subadditivity, which asserts that
  \[ S(A_1 \dots A_b) + S(A_2 \dots A_a) \geq S(A_2 \dots A_b) + S(A_1 \dots A_a), \]
  or
  \[ S(A_2 \dots A_a) - S(A_2 \dots A_b) \geq S(A_1 \dots A_a) - S(A_1 \dots A_b). \]
  By summing over all the cyclic shifts and subtracting $S(A_1 \dots A_n)$ from both sides of the resulting inequality we obtain that $C_n(a-1,b-1) \geq C_n(a,b)$.

  The second claim is likewise a direct consequence of weak monotonicity, which shows that
  \[ S(A_1 \dots A_{a+1}) + S(A_{a+1} \dots S_{a+b}) \geq S(A_1 \dots A_a) + S(A_{a+2} \dots S_{a+b}), \]
  or
  \[ S(A_1 \dots A_{a+1}) - S(A_{a+2} \dots A_{a+b}) \geq S(A_1 \dots A_a) - S(A_{a+1} \dots S_{a+b}). \]

  The third claim is a reformulation of the second for the case $b = a + 1$.
\end{proof}

Now consider a general instance of \cref{eq:general cyclic appendix} that satisfies the necessary conditions of \cref{lem:cyclic necessary}.
Then we may apply the first part of \cref{lem:reduction} to reduce to the case where $a+b \in \{n-1,n\}$.
The second part can be used further to reduce to $a-b \in \{1,2\}$.
Thus it suffices to prove the following four instances of \cref{eq:general cyclic appendix} for arbitrary $k$:
\begin{equation}
\label{eq:final four}
  C_{2k+1}(k+1,k), \quad
  C_{2k+1}(k+1,k-1), \quad
  C_{2k}(k+1,k-1), \quad 
  C_{2k}(k,k-1).
\end{equation}
In fact, the following lemma allows us to focus on a single case.

\begin{lem}
\label{lem:reduction two}
  It suffices to prove the cyclic inequalities $C_{2k+1}(k+1,k)$ for all $k$.
\end{lem}
\begin{proof}
  We show that all other instances in \cref{eq:final four} can be reduced to instances of the first one.
  For the second instance, this follows simply by applying the third reduction in \cref{lem:reduction}:
  \begin{align*}
    &C_{2k+1}(k+1,k-1)
    = C_{2k+1}(k+1) - C_{2k+1}(k-1) - S(A_1 \dots A_n) \\
    \geq\;&C_{2k+1}(k+1) - C_{2k+1}(k) - S(A_1 \dots A_n)
    = C_{2k+1}(k+1,k)
  \end{align*}

  For the remaining two cases, where the number $n=2k$ of regions is even, this is somewhat more involved and we will use the following trick:
  For each $i \in [2k]$, consider the entropy inequality that is obtained by applying $C_{2k-1}(k,k-1)$ to the
  $n-1$ regions obtained by combining region $A_i$ and $A_{i+1}$.
  Let us denote this inequality by $C^{(i)}_{2k}$. By summing over all $i$, we find that
  \begin{align*}
    \sum_{i=1}^{2k} C^{(i)}_{2k} = &\big( k \, C_{2k}(k+1) + (k-1) C_{2k}(k) \big) - \big( (k-1) C_{2k}(k) - k \, C_{2k}(k-2) \big) \\
    -\;&2k \, S(A_1 \dots A_{2k}) \geq 0,
  \end{align*}
  or
  \begin{equation}
  \label{eq:strong even}
    \frac {C_{2k}(k+1) - C_{2k}(k-1)} 2 \geq S(A_1 \dots A_{2k}).
  \end{equation}
  This immediately implies the third inequality $C_{2k}(k+1,k-1) \geq 0$.
  The fourth inequality, $C_{2k}(k,k-1) \geq 0$, can likewise be obtained from \cref{eq:strong even} by using
  \[ C_{2k}(k) - C_{2k}(k-1) \geq \frac {C_{2k}(k+1) - C_{2k}(k-1)} 2, \]
  which is a reformulation of the first reduction in \cref{lem:reduction}.
\end{proof}

\paragraph{Proof by contraction.}
We now construct contractions $f_k \colon \{0,1\}^{2k+1} \rightarrow \{0,1\}^{2k+2}$ for each of the strongest cyclic inequalities $C_{2k+1}(k+1,k)$ that we had reduced to in \cref{lem:reduction two}.

We first define a map $g_k \colon \{0,1\}^{2k+1} \setminus \{0\} \rightarrow \{0,1\}^{2k+1}$ on all non-zero bitstrings in the following way:
The map $g_k$ preserves all ones and acts on each block of consecutive zeros in the following way:
\begin{itemize}
  \item If the length of the block of zeros is even, it is mapped to all ones:
  \[ \cdots\underline{00\cdots0}\cdots \mapsto \cdots\underline{11\cdots1}\cdots \]
  \item If the length of the block of zeros is odd, all digits except for the first are mapped to ones:
  \[ \cdots\underline{00\cdots0}\cdots \mapsto \cdots\underline{01\cdots1}\cdots \]
\end{itemize}
We allow for blocks of zeros to cyclically wrap around.
As we only consider non-zero bitstrings, the first digit in a block of zeros is always well-defined.
The following lemma is key to our argument:

\begin{lem}
\label{lem:cyclic contraction}
  The map $g_k$ is a contraction.
\end{lem}
\begin{proof}
  We need to show that $\norm{g_k(x) - g_k(x')}_1 \leq \norm{x - x'}_1$ for all non-zero bitstrings $0 \neq x, x' \in \{0,1\}^{2k+1}$.
  We first consider the case where $\norm{x-x'}_1 = 1$.
  Without loss of generality, assume that $\norm{x'}_1 = \norm{x} + 1$.
  That is, $x'$ can be obtained from $x$ by flipping a single zero to one.
  Let $L$ denote the number of zeros on its left and $R$ the number of zeros on its right, so that $N=L+R+1$ is the total number of zeros in the block:
  \begin{equation}
  \label{eq:blocks of zeros}
     x  = \cdots \underline{0^L 0 0^R} \cdots \quad \text{ and } \quad
     x' = \cdots \underline{0^L 1 0^R} \cdots .
  \end{equation}
  Thus the block of zeros in $x$ is split into (generically) two blocks of zeros in $x'$; all other blocks of zeros are untouched.
  As the map $g_k$ acts independently on each block of zeros, it therefore suffices to consider its action on the blocks of zeros displayed in \cref{eq:blocks of zeros}.
  We now show that $\norm{g_k(x) - g_k(x')}_1 \leq 1$ by distinguishing four cases:
  \begin{itemize}
    \item The flipped zero is contained in a block of even length $N$:
    \begin{itemize}
      \item $L$ is even and $R$ is odd:
      \[
        \norm{g_k(\cdots \underline{0^L 0 0^R} \cdots) - g_k(\cdots \underline{0^L 1 0^R} \cdots)}_1 =
        \norm{\cdots \underline{1^L 1 1^R} \cdots - \cdots \underline{1^L 1 01^{R-1}} \cdots}_1 =
        1
      \]
      \item $L$ is odd and $R$ is even:
      \[
        \norm{g_k(\cdots \underline{0^L 0 0^R} \cdots) - g_k(\cdots \underline{0^L 1 0^R} \cdots)}_1 =
        \norm{\cdots \underline{1^L 1 1^R} \cdots - \cdots \underline{01^{L-1} 1 1^R} \cdots}_1 =
        1
      \]
    \end{itemize}

    \item The flipped zero is contained in a block of odd length $N$:
    \begin{itemize}
      \item $L$ and $R$ are both odd:
      \[
        \norm{g_k(\cdots \underline{0^L 0 0^R} \cdots) - g_k(\cdots \underline{0^L 1 0^R} \cdots)}_1 =
        \norm{\cdots \underline{01^{L-1} 1 1^R} \cdots - \cdots \underline{01^{L-1} 1 01^{R-1}} \cdots}_1 =
        1
      \]
      \item $L$ and $R$ are both even:
      \[
        \norm{g_k(\cdots \underline{0^L 0 0^R} \cdots) - g_k(\cdots \underline{0^L 1 0^R} \cdots)}_1 =
        \norm{\cdots \underline{01^L 1^R} \cdots - \cdots \underline{1^L 1 1^R} \cdots}_1 =
        1
      \]
    \end{itemize}
  \end{itemize}
  To conclude the proof, observe that the general case where $\norm{x-x'}_1 = d$ can always be reduced to the above.
  For this, choose a sequence of $d+1$ non-zero bitstrings $x_0, \dots, x_d$ such that $x_0 = x$, $x_d = x'$ and each $\norm{x_i - x_{i+1}}_1 = 1$.
  Then the triangle inequality and the above show that, indeed,
  \[
    \norm{g_k(x) - g_k(x')}_1
    \leq \sum_{i=0}^{d-1} \norm{g_k(x_i) - g_k(x_{i+1})}_1
    \leq \sum_{i=0}^{d-1} \norm{x_i - x_{i+1}}_1
    = \norm{x - x'}_1.
    \qedhere
  \]
\end{proof}

We now define the contraction that will prove the extremal cyclic inequality:
\begin{equation}
\label{eq:cyclic contraction}
  f_k \colon \begin{cases}
    \{0,1\}^{2k+1} &\rightarrow \{0,1\}^{2k+1} \times \{0,1\} = \{0,1\}^{2k+2} \\
    0 &\mapsto (0, 0) \\
    x &\mapsto (\overline{\pi_k(g_k(\pi_k^{-1}(x)))}, 1)
  \end{cases}.
\end{equation}
Here, $\overline{\cdots}$ denotes the bitwise inversion and $\pi_k$ denotes the permutation operator that acts on bitstrings of length $2k+1$ in the following way:
\[ \pi_k (v_0, v_1, \dots, v_{2k}) = (v_0, v_2, \dots, v_{2k}, v_1, v_3, \dots, v_{2k-1}). \]
The following theorem and corollary conclude our proof of the cyclic entropy inequalities:

\begin{thm}
\label{thm:strongest cyclic}
  The map $f_k$ is a `proof by contraction' of the cyclic entropy inequality $C_{2k+1}(k+1,k) \geq 0$.
\end{thm}
\begin{proof}
  We first argue that $f_k$ is a contraction.
  As in the proof of \cref{lem:cyclic contraction}, it suffices to show that $\norm{f_k(x) - f_k(x')}_1 \leq 1$ for $\norm{x'} = \norm{x} + 1$.
  But since the permutation $\pi_k$ and the bitwise inversion are isometries, this is already proved in all cases except when $x = 0$ and $x'$ contains a single one.
  In the latter case, $\pi_k^{-1}(x')$ likewise contains only a single one.
  In other words, there is a single block of $2k$ zeros, so that $g_k(\pi_k^{-1}(x'))$ is the bitstring of all ones.
  It follows that $f_k(x') = (0,1)$, which has Hamming distance one from $f_k(x) = (0,0)$.
  Therefore $f_k$ is indeed a contraction.

  It remains to show that $f_k$ maps the occurrence vectors \cref{eq:occurrence vectors} onto each other.
  To fix our conventions, we shall write $C_{2k+1}(k+1,k) \geq 0$ in the form
  \begin{equation}
  \label{eq:cyclic convention}
    \sum_{i=1}^{2k+1} S(A_{i+k+1} \dots A_{i+2k+1}) \geq \sum_{i=1}^{2k+1} S(A_{i+1} \dots A_{i+k}) + S(A_1 \dots A_{2k+1}).
  \end{equation}
  Then the occurrence vectors of the $2k+1$ regions are given by $x_1 = 1^{k+1}0^k$, $y_1 = (0^{k+1}1^k, 1)$ and their respective cyclic permutations.
  One readily verifies that indeed $f_k(x_i) = y_i$ for all $i \in [2k+1]$.
  Indeed, each $\pi_k^{-1}(x_i)$ is a fixed point of $g_k$ and so $f_k$ acts as $f_k(x_i) = (\overline{x_i},1) = y_i$.
  On the other hand, for the purifying region $f_k(0) = (0,0)$ holds by definition.
\end{proof}

In view of \cref{lem:reduction two}, we obtain the following corollary as an immediate consequence.

\begin{cor}
  The cyclic entropy inequality $C_n(a,b) \geq 0$ is a valid holographic entropy inequality if and only if $a>b$ and $a+b \leq n$.
\end{cor}

We illustrate the construction in the case of five regions ($k=2$). Here, the cyclic inequality \cref{eq:cyclic convention} reads
\begin{equation}
\label{eq:cyclic five appendix}
\begin{aligned}
  &S(DEA) + S(EAB) + S(ABC) + S(BCD) + S(CDE) \\
  \geq\; &S(BC) + S(CD) + S(DE) + S(EA) + S(AB) + S(ABCDE).
\end{aligned}
\end{equation}
The contraction $f_2$ as defined in \cref{eq:cyclic contraction} is displayed in \cref{tab:cyclic-contraction}.

\begin{table}
\centering
\begin{tabular}{c ccccc ccccc c}
  \toprule
  & \multicolumn{5}{c}{$x$} & \multicolumn{6}{c}{$y = f_2(x)$} \\
  \cmidrule(lr){2-6} \cmidrule(lr){7-12}
  & DEA & EAB & ABC & BCD & CDE & BC & CD & DE & EA & AB & ABCDE \\
O & 0 & 0 & 0 & 0 & 0 & 0 & 0 & 0 & 0 & 0 & 0 \\
  & 0 & 0 & 0 & 0 & 1 & 0 & 0 & 0 & 0 & 0 & 1 \\
  & 0 & 0 & 0 & 1 & 0 & 0 & 0 & 0 & 0 & 0 & 1 \\
  & 0 & 0 & 0 & 1 & 1 & 0 & 1 & 0 & 0 & 0 & 1 \\
  & 0 & 0 & 1 & 0 & 0 & 0 & 0 & 0 & 0 & 0 & 1 \\
  & 0 & 0 & 1 & 0 & 1 & 1 & 0 & 0 & 0 & 0 & 1 \\
  & 0 & 0 & 1 & 1 & 0 & 1 & 0 & 0 & 0 & 0 & 1 \\
C & 0 & 0 & 1 & 1 & 1 & 1 & 1 & 0 & 0 & 0 & 1 \\
  & 0 & 1 & 0 & 0 & 0 & 0 & 0 & 0 & 0 & 0 & 1 \\
  & 0 & 1 & 0 & 0 & 1 & 0 & 0 & 1 & 0 & 0 & 1 \\
  & 0 & 1 & 0 & 1 & 0 & 0 & 0 & 0 & 0 & 1 & 1 \\
  & 0 & 1 & 0 & 1 & 1 & 0 & 0 & 0 & 0 & 0 & 1 \\
  & 0 & 1 & 1 & 0 & 0 & 0 & 0 & 0 & 0 & 1 & 1 \\
  & 0 & 1 & 1 & 0 & 1 & 0 & 0 & 0 & 0 & 0 & 1 \\
B & 0 & 1 & 1 & 1 & 0 & 1 & 0 & 0 & 0 & 1 & 1 \\
  & 0 & 1 & 1 & 1 & 1 & 1 & 0 & 0 & 0 & 0 & 1 \\
  & 1 & 0 & 0 & 0 & 0 & 0 & 0 & 0 & 0 & 0 & 1 \\
  & 1 & 0 & 0 & 0 & 1 & 0 & 0 & 1 & 0 & 0 & 1 \\
  & 1 & 0 & 0 & 1 & 0 & 0 & 1 & 0 & 0 & 0 & 1 \\
D & 1 & 0 & 0 & 1 & 1 & 0 & 1 & 1 & 0 & 0 & 1 \\
  & 1 & 0 & 1 & 0 & 0 & 0 & 0 & 0 & 1 & 0 & 1 \\
  & 1 & 0 & 1 & 0 & 1 & 0 & 0 & 0 & 0 & 0 & 1 \\
  & 1 & 0 & 1 & 1 & 0 & 0 & 0 & 0 & 0 & 0 & 1 \\
  & 1 & 0 & 1 & 1 & 1 & 0 & 1 & 0 & 0 & 0 & 1 \\
  & 1 & 1 & 0 & 0 & 0 & 0 & 0 & 0 & 1 & 0 & 1 \\
E & 1 & 1 & 0 & 0 & 1 & 0 & 0 & 1 & 1 & 0 & 1 \\
  & 1 & 1 & 0 & 1 & 0 & 0 & 0 & 0 & 0 & 0 & 1 \\
  & 1 & 1 & 0 & 1 & 1 & 0 & 0 & 1 & 0 & 0 & 1 \\
A & 1 & 1 & 1 & 0 & 0 & 0 & 0 & 0 & 1 & 1 & 1 \\
  & 1 & 1 & 1 & 0 & 1 & 0 & 0 & 0 & 1 & 0 & 1 \\
  & 1 & 1 & 1 & 1 & 0 & 0 & 0 & 0 & 0 & 1 & 1 \\
  & 1 & 1 & 1 & 1 & 1 & 0 & 0 & 0 & 0 & 0 & 1 \\
\bottomrule
\end{tabular}
\captionof{table}{Proof by contraction of the five-party cyclic entropy inequality \cref{eq:cyclic five appendix}.}
\label{tab:cyclic-contraction}
\end{table}

The corresponding cuts for the right-hand side of the inequality formed according to $f_2$ from the minimal cuts $W_l$ for the left-hand side are given as follows:
\begin{align*}
U_{BC} = & W(00111) \cup W(00101) \cup W(01110) \cup W(00110) \cup W(01111) \\
= & W_{DEA} \cap W_{EAB} \cap W_{ABC}^c \cap W_{BCD}^c \cap W_{CDE}^c \\ \cup &W_{DEA} \cap W_{EAB} \cap W_{ABC}^c \cap W_{BCD} \cap W_{CDE}^c \\ \cup &W_{DEA} \cap W_{EAB}^c \cap W_{ABC}^c \cap W_{BCD}^c \cap W_{CDE} \\ \cup &W_{DEA} \cap W_{EAB} \cap W_{ABC}^c \cap W_{BCD}^c \cap W_{CDE} \\ \cup &W_{DEA} \cap W_{EAB}^c \cap W_{ABC}^c \cap W_{BCD}^c \cap W_{CDE}^c, \\
U_{CD} = & W(00111) \cup W(10010) \cup W(10111) \cup W(00011) \cup W(10011) \\
= & W_{DEA} \cap W_{EAB} \cap W_{ABC}^c \cap W_{BCD}^c \cap W_{CDE}^c \\ \cup &W_{DEA}^c \cap W_{EAB} \cap W_{ABC} \cap W_{BCD}^c \cap W_{CDE} \\ \cup &W_{DEA}^c \cap W_{EAB} \cap W_{ABC}^c \cap W_{BCD}^c \cap W_{CDE}^c \\ \cup &W_{DEA} \cap W_{EAB} \cap W_{ABC} \cap W_{BCD}^c \cap W_{CDE}^c \\ \cup &W_{DEA}^c \cap W_{EAB} \cap W_{ABC} \cap W_{BCD}^c \cap W_{CDE}^c, \\
U_{DE} = & W(11011) \cup W(10001) \cup W(10011) \cup W(01001) \cup W(11001) \\
= & W_{DEA}^c \cap W_{EAB}^c \cap W_{ABC} \cap W_{BCD}^c \cap W_{CDE}^c \\ \cup &W_{DEA}^c \cap W_{EAB} \cap W_{ABC} \cap W_{BCD} \cap W_{CDE}^c \\ \cup &W_{DEA}^c \cap W_{EAB} \cap W_{ABC} \cap W_{BCD}^c \cap W_{CDE}^c \\ \cup &W_{DEA} \cap W_{EAB}^c \cap W_{ABC} \cap W_{BCD} \cap W_{CDE}^c \\ \cup &W_{DEA}^c \cap W_{EAB}^c \cap W_{ABC} \cap W_{BCD} \cap W_{CDE}^c, \\
U_{EA} = & W(11100) \cup W(10100) \cup W(11101) \cup W(11000) \cup W(11001) \\
= & W_{DEA}^c \cap W_{EAB}^c \cap W_{ABC}^c \cap W_{BCD} \cap W_{CDE} \\ \cup &W_{DEA}^c \cap W_{EAB} \cap W_{ABC}^c \cap W_{BCD} \cap W_{CDE} \\ \cup &W_{DEA}^c \cap W_{EAB}^c \cap W_{ABC}^c \cap W_{BCD} \cap W_{CDE}^c \\ \cup &W_{DEA}^c \cap W_{EAB}^c \cap W_{ABC} \cap W_{BCD} \cap W_{CDE} \\ \cup &W_{DEA}^c \cap W_{EAB}^c \cap W_{ABC} \cap W_{BCD} \cap W_{CDE}^c, \\
U_{AB} = & W(11100) \cup W(11110) \cup W(01110) \cup W(01100) \cup W(01010) \\
= & W_{DEA}^c \cap W_{EAB}^c \cap W_{ABC}^c \cap W_{BCD} \cap W_{CDE} \\ \cup &W_{DEA}^c \cap W_{EAB}^c \cap W_{ABC}^c \cap W_{BCD}^c \cap W_{CDE} \\ \cup &W_{DEA} \cap W_{EAB}^c \cap W_{ABC}^c \cap W_{BCD}^c \cap W_{CDE} \\ \cup &W_{DEA} \cap W_{EAB}^c \cap W_{ABC}^c \cap W_{BCD} \cap W_{CDE} \\ \cup &W_{DEA} \cap W_{EAB}^c \cap W_{ABC} \cap W_{BCD}^c \cap W_{CDE}, \\
U_{ABCDE} = & W(00000)^c = W_{DEA} \cup W_{EAB} \cup W_{ABC} \cup W_{BCD} \cup W_{CDE}.
\end{align*}
We remark that this proof by contraction is essentially unique (up to trivial choices).

\paragraph{Independence.}
We finally show that each cyclic inequality $C_{2k+1}(k+1,k)$ is independent from all instances with smaller $k$.

\begin{prp}
\label{prp:independence}
  For each odd $n=2k+1$, there exists a quantum state $\rho_n$ whose von Neumann entropies violate $C_{2k+1}(k+1,k)$ but none of the cyclic entropy inequalities $C_{2l+1}(l+1,l)$ with $l<k$ (nor their permutations).

  Therefore, each of the cyclic entropy inequalities $C_{2k+1}(k+1,k)$ is independent from those with smaller $k$ (and their permutations).
\end{prp}
\begin{proof}
  Let $A$ be a matrix of size $(n+1) \times k$ with entries in a finite field $\mathbb F_p$ such that all subsets of rows have maximal rank.
  It is easy to see that such a matrix exists if $p$ is large enough.
  Now consider the pure state
  $\ket{\psi_{n+1}} = \frac 1 Z \sum_{x \in \mathbb F_p^k} \ket{A x} \in (\mathbb C^p)^{\otimes n+1}$,
  where we write $\ket y = \ket{y_1} \otimes \dots \otimes \ket{y_{n+1}}$ for any $y \in \mathbb F_p^{n+1}$ and where $Z$ is a suitable normalization constant.
  We remark that this family of states generalizes the four-party GHZ state.
  It is not hard to see that the von Neumann entropies associated with $\ket{\psi_{n+1}}$ are proportional to
  \begin{equation}
  \label{eq:special entropies}
    S(I) =  \min \{ \abs I, n+1 - \abs I, k \}
  \end{equation}
  by an overall factor of $\log p$.
  The same is true for its $n$-body reduced state $\rho_n$, which by construction is fully $S_{n+1}$-permutation symmetric. 

  We now evaluate the cyclic entropy inequalities.
  On the one hand, it is easy to see from \cref{eq:special entropies} that
  \[ C_{2k+1}(k+1,k) = (2k+1)(k-k) - 1 = -1 < 0.\]
  Therefore $\rho_n$ violates the cyclic inequality $C_{2k+1}(k+1,k)$.
  On the other hand, none of the cyclic inequalities $C_{2l+1}(l+1,l)$ with $l<k$ are violated by $\rho_n$.
  Indeed, we find that
  \[ C_{2l+1}(l+1,l) = (2l+1)(l+1-l) - S(A_1 \dots A_{2l+1}) \geq (2l + 1) - (2l + 1) = 0 \]
  when applied to the first $2l+1$ parties.
  While there are in general many more ways of applying an inequality for a fewer number of regions to a state of $n=2k+1$ regions (see \cref{subsec:symmetries}) they are in this case all equivalent as the entropies \cref{eq:special entropies} are fully $S_{n+1}$-symmetric.
\end{proof}

\bibliographystyle{JHEP}
\bibliography{holographic_cone}

\end{document}